\documentclass[10pt]{article}
\usepackage{geometry}
\usepackage{bm}
\usepackage{amsthm}
\usepackage{amsmath}
\usepackage{amssymb}
\usepackage{graphicx}
\usepackage[unicode=true,pdfusetitle,
 bookmarks=true,bookmarksnumbered=false,bookmarksopen=false,
 breaklinks=false,pdfborder={0 0 1},backref=false,colorlinks=false]
 {hyperref}
\usepackage{breakurl}
\usepackage{appendix}

\makeatletter
  \theoremstyle{plain}
  \newtheorem{thm}{\protect\theoremname}
  \theoremstyle{plain}
  \newtheorem{prop}{\protect\propositionname}
  \theoremstyle{plain}
  \newtheorem{cor}{\protect\corollaryname}
  \theoremstyle{plain}
  \newtheorem{lem}{\protect\lemmaname}
 \ifx\proof\undefined\
   \newenvironment{proof}[1][\proofname]{\par
     \normalfont\topsep6\p@\@plus6\p@\relax
     \trivlist
     \itemindent\parindent
     \item[\hskip\labelsep
           \scshape
       #1]\ignorespaces
   }{%
     \endtrivlist\@endpefalse
   }
   \providecommand{\proofname}{Proof}
 \fi

\@ifundefined{date}{}{\date{}}

\makeatother

\providecommand{\corollaryname}{Corollary}
\providecommand{\lemmaname}{Lemma}
\providecommand{\propositionname}{Proposition}
\providecommand{\theoremname}{Theorem}

\begin{document}
\global\long\def\P{\mathcal{P}}
 \global\long\def\R{\mathcal{R}}
 \global\long\def\II{\mathcal{\mathcal{I}}}

\global\long\def\M{\bm{M}}
 \global\long\def\X{\bm{X}}
 \global\long\def\Y{\bm{Y}}
 \global\long\def\Z{\bm{Z}}
 \global\long\def\U{\bm{U}}
 \global\long\def\V{\bm{V}}
 \global\long\def\W{\bm{W}}
 \global\long\def\A{\bm{A}}

\global\long\def\bu{\bar{\bm{U}}}
\global\long\def\bv{\bar{\V}}
 \global\long\def\D{\bm{D}}
 \global\long\def\E{\bm{E}}
 \global\long\def\F{\bm{F}}
 \global\long\def\G{\bm{G}}
 \global\long\def\H{\bm{H}}
 \global\long\def\I{\bm{I}}
 \global\long\def\S{\bm{S}}
 \global\long\def\L{\bm{L}}

\global\long\def\e{\bm{e}}
 \global\long\def\x{\bm{x}}
 \global\long\def\y{\bm{y}}
 \global\long\def\z{\bm{z}}
 \global\long\def\u{\bm{u}}
 \global\long\def\v{\bm{v}}
 \global\long\def\w{\bm{w}}
 \global\long\def\s{\bm{s}}

\global\long\def\ei{\bm{\e}_{i}}
 \global\long\def\ej{\bm{e}_{j}}

\global\long\def\bDelta{\boldsymbol{\Delta}}

\title{Incoherence-Optimal Matrix Completion}

\author{Yudong Chen\\
Department of Electrical Engineering and Computer Sciences\\
The University of California, Berkeley\\
yudong.chen@eecs.berkeley.edu}

\maketitle
 
\begin{abstract}
This paper considers the matrix completion problem. We show that it
is not necessary to assume \emph{joint incoherence}, which is a standard
but unintuitive and restrictive condition that is imposed by previous
studies. This leads to a sample complexity bound that is order-wise
optimal with respect to the incoherence parameter (as well as to the
rank $r$ and the matrix dimension $n$ up to a log factor). As a
consequence, we improve the sample complexity of recovering a semidefinite
matrix from $O(nr^{2}\log^{2}n)$ to $O(nr\log^{2}n)$, and the highest
allowable rank from $\Theta(\sqrt{n}/\log n)$ to $\Theta(n/\log^{2}n)$.
The key step in proof is to obtain new bounds on the $\ell_{\infty,2}$-norm,
defined as the maximum of the row and column norms of a matrix. To
illustrate the applicability of our techniques, we discuss extensions
to SVD projection, structured matrix completion and semi-supervised
clustering, for which we provide order-wise improvements over existing
results. Finally, we turn to the closely-related problem of low-rank-plus-sparse
matrix decomposition. We show that the joint incoherence condition
is unavoidable here for polynomial-time algorithms conditioned on
the Planted Clique conjecture. This means it is intractable in general
to separate a rank-$\omega(\sqrt{n})$ positive semidefinite matrix
and a sparse matrix. Interestingly, our results show that the standard
and joint incoherence conditions are associated respectively with
the \emph{information} (statistical) and \emph{computational} aspects of the matrix
decomposition problem.
\end{abstract}

\section{Introduction}

The matrix completion problem concerns recovering a low-rank matrix
from an observed subset of its entries. Recent research~\cite{candes2010NearOptimalMC,Recht2009SimplerMC,gross2009anybasis,keshavan2009matrixafew,jain2013altMin}
has demonstrated the following remarkable fact: if a rank-$r$ $n\times n$
matrix satisfies certain incoherence properties, then it is possible
to exactly reconstruct the matrix with high probability from $nr\mbox{polylog}(n)\ll n^{2}$
uniformly sampled entries using efficient polynomial-time algorithms.

In previous work, the sample complexity $\Theta(nr\mbox{polylog}(n))$
is achieved only for matrices that satisfy two types of incoherence
conditions with constant parameters. The first condition, known as
\emph{standard incoherence}, is a natural and necessary requirement;
it prevents the information of the row and column spaces of the matrix
from being too concentrated in a few rows or columns. A second condition,
called \emph{joint incoherence} (or strong incoherence), is also needed.
It requires the left and right singular vectors of the matrix to be
unaligned with each other. This condition is quite unintuitive, and
does not seem to have a natural interpretation. As we demonstrate
later, this condition is often restrictive and precludes a large class
of otherwise well-conditioned matrices. For example, positive semidefinite
matrices have a non-constant joint incoherence parameter on the order
of $\Omega(r)$, and previous results thus require the number of observations
to be proportional to $nr^{2}$ instead of $nr$. In several applications
of matrix completion discussed later, the joint incoherence condition
leads to artificial and undesired constraints. In contrast, numerical
experiments suggest that this condition is not needed.

In this paper, we prove that \emph{the joint incoherence condition
is not necessary} and can be completely eliminated. With $\Omega(nr\log^{2}n)$
uniformly sampled entries, one can recover a matrix that satisfies
the standard incoherence condition (with a constant parameter) but
is not jointly incoherent (e.g., a positive semidefinite matrix).
As we show in Section~\ref{sec:main_strong}, our sample complexity
bounds are order-wise optimal with respect to not only the matrix
dimensions $n$ and $r$ but also to its incoherence parameters except
for a $\log n$ factor. As a consequence, we improve the sample complexity
of recovering a positive semidefinite matrix from $O(nr^{2}\log^{2}n)$
to $O(nr\log^{2}n)$, and the highest allowable rank from $\Theta(\sqrt{n}/\log n)$
to $\Theta(n/\log^{2}n)$.

Our results apply to the standard nuclear norm minimization approach
to matrix completion. The improvements are achieved by a new analysis
based on bounds involving the $\ell_{\infty,2}$ matrix norm, defined
as the maximum of the row and column norms of the matrix. This differs
from previous approaches that use $\ell_{\infty}$ bounds. We show
that this technique can be extended to obtain strong theoretical guarantees
for the following two problems:
\begin{enumerate}
\item the analysis of a Singular Value Decomposition (SVD) projection algorithm for matrix completion;
\item structured matrix completion and semi-supervised clustering with side
information.
\end{enumerate}
In both problems we achieve order-wise improvements over existing
results. We believe the $\ell_{\infty,2}$ norm is useful more broadly. For example, in the follow-up work~\cite{chen2013coherent},
a weighted version of the $\ell_{\infty,2}$ norm plays a crucial
role in the analysis of general low-rank matrices that violate the
standard incoherence condition.

Finally, we turn to the closely related problem of matrix decomposition,
where one is asked to recover a low-rank matrix and a sparse matrix
from their sum. We show that the joint incoherence condition is necessary
in this setting based on the computational complexity assumption of
the Planted Clique problem. In particular, any decomposition algorithm
that does not require the joint incoherence condition would solve
Planted Clique with clique size $o(\sqrt{n})$, a problem that has
been extensively studied and is widely believed to be intractable
in polynomial time. This implies that it is computationally hard in
general to separate a rank-$\omega(\sqrt{n})$ positive semidefinite
matrix and a sparse matrix. Interestingly, our results show that the
standard incoherence condition is inherently associated the \emph{information-theoretic}
(or \emph{statistical}) aspect of the problem, whereas the joint incoherence
condition reflects the \emph{computational} aspect.

\paragraph*{Related work}

We briefly survey existing related work; detailed comparisons with
these results are provided after we present our theorems. Matrix completion
is first studied in~\cite{candes2009exact}, which initiates the
use of the nuclear norm minimization approach. The work in \cite{candes2010NearOptimalMC,Recht2009SimplerMC,gross2009anybasis,keshavan2009matrixafew}
provides state-of-the-art theoretical guarantees on exact completion.
Alternative algorithms for matrix completion are considered in~\cite{jain2013altMin,keshavan2009matrixafew,cai2010singular}.
All these works require the joint incoherence condition (or 
a sample complexity that is at least quadratic in $r$). Our extensions to
SVD projection, structured matrix completion and semi-supervised clustering
are inspired by the work in~\cite{keshavan2009matrixafew,yi2013assisted};
we improve upon their results. The low-rank and sparse matrix decomposition
problem is considered first in~\cite{chandrasekaran2011siam,candes2009robustPCA}
and subsequently in~\cite{chen2011LSarxiv,li2013constantCorruption,agarwal2012decomposition,Hsu2010RobustDecomposition}.
The work in~\cite{candes2009robustPCA,chen2011LSarxiv,li2013constantCorruption}
prove the success of specific algorithms assuming the standard and
joint incoherence conditions. Our results show that these two incoherence
conditions are in fact necessary for \emph{all} algorithms, or all
polynomial-time algorithms, due to statistical and computational reasons.
The seminal work in~\cite{berthet2012sparsePCA,berthet2013colt_sparsePCA}
is the first to use the Planted Clique problem to establish statistical
limits under computational constraints; they consider the problem
of sparse Principal Component Analysis (PCA). A similar approach is
taken in~\cite{ma2013submatrix} for the submatrix detection problem.

\paragraph*{Organization}

In Section~\ref{sec:main_strong} we present our main result and
show that the joint incoherence condition is not needed in matrix
completion. We discuss extensions to SVD projection and structured
matrix completion in Section~\ref{sec:extension}. In Section~\ref{sec:LS}
we turn to the matrix decomposition problem and show that the joint
incoherence condition is unavoidable there. We prove our main theorem
in Section~\ref{sec:Proofs}, with some technical aspects of the
proofs deferred to the appendix. The paper is concluded with a discussion
in Section~\ref{sec:discussion}.

\paragraph*{Notation}

Lower case bold letters (e.g., $\z$) denote vectors, while capital
bold face letters (e.g., $\Z$) denote matrices. For a matrix $\Z$,
$Z_{ij}$ and $(\Z)_{ij}$ both denote its $(i,j)$-th entry. By \emph{with
high probability} (\emph{w.h.p.}) we mean with probability at least
$1-c_{1}(n_{1}+n_{2})^{-c_{2}}$ for some universal constants $c_{1},c_{2}>0$,
where $n_{1}$ and $n_{2}$ are dimensions of the low-rank matrix. $ \Vert \cdot \Vert_2 $ denotes the vector $ \ell_2 $ norm. $ \Vert \cdot \Vert_F $, $ \Vert \cdot \Vert_* $ and $ \Vert \cdot \Vert $ denote the Frobenius, nuclear and spectral norms for matrices, respectively.

\section{Main Results\label{sec:main_strong}}

We now define the matrix completion problem. Suppose $\M\in\mathbb{R}^{n_{1}\times n_{2}}$
is an unknown matrix with rank~$r$. For each $(i,j)$, $M_{ij}$
is observed with probability $p$ independent of all others.%
\footnote{\label{fn:sampling_model}This is known as the Bernoulli model~\cite{candes2009robustPCA}.
Other widely used models include the sampling with/without replacement
models~\cite{gross2009anybasis,candes2009exact,gross2010noreplacement,Recht2009SimplerMC}.
Recovery guarantees for one model can be easily translated to others
with only a change in constant factors~\cite{candes2010NearOptimalMC,gross2010noreplacement}.%
} Let $\Omega$ be the set of the indices of the observed entries.
The matrix completion problem asks for recovering $\M$ from the observations
$\left\{ M_{ij,}(i,j)\in\Omega\right\} $. The standard and arguably
the most popular approach to matrix completion is the nuclear norm
minimization method~\cite{candes2009exact}: 
\begin{equation}
\begin{split}\min_{X} & \quad\left\Vert \X\right\Vert _{*}\\
\mbox{s.t.} & \quad X_{ij}=M_{ij}\textrm{ for \ensuremath{(i,j)\in\Omega},}
\end{split}
\label{eq:cvx_opt}
\end{equation}
where $\left\Vert \X\right\Vert _{*}$ is the nuclear norm of the
matrix $\X$, defined as the sum of its singular values. Our goal
is to obtain sufficient conditions under which the optimal solution
to the problem~(\ref{eq:cvx_opt}) is unique and equal to $\M$ with
high probability.

It is observed in~\cite{candes2009exact} that if $\M$ is equal
to zero in nearly all of rows or columns, then it is impossible to
complete $\M$ unless all of its entries of are observed. To avoid
such pathological situations, it has become standard to assume $\M$
to have additional properties known as incoherence. Suppose the rank-$r$
SVD of $\M$ is $\U\mathbf{\Sigma}\V^{\top}$. $\M$ is said to satisfy
the \emph{standard incoherence} condition with parameter $\mu_{0}$
if 
\begin{align}
\begin{split}\max_{1\le i\le n_{1}}\left\Vert \U^{\top}\ei\right\Vert _{2} & \le\sqrt{\frac{\mu_{0}r}{n_{1}}},\\
\max_{1\le j\le n_{2}}\left\Vert \V^{\top}\ej\right\Vert _{2} & \le\sqrt{\frac{\mu_{0}r}{n_{2}}},
\end{split}
\label{eq:incoherence}
\end{align}
where $\ei$ are the $i$-th standard basis with appropriate dimension.
Note that $1\le\mu_{0}\le \frac{\min \{n_{1},n_{2}\}}{r}$. Previous work also
requires $\M$ to satisfy an additional \emph{joint incoherence} (or
\emph{strong incoherence}) condition with parameter $\mu_{1}$, defined
as 
\begin{equation}
\max_{i,j}\left|\left(\U\V^{\top}\right)_{ij}\right|\le\sqrt{\frac{\mu_{1}r}{n_{1}n_{2}}}.\label{eq:strong_incoherence}
\end{equation}
Under these two conditions, existing results require $p\gtrsim\max\{\mu_{0},\mu_{1}\}r\textrm{polylog}(n)/n$
to recover $\M\in\mathbb{R}^{n\times n}$. If we let $\mu_{0}$ and
$\mu_{1}$ to be the smallest numbers that satisfy~(\ref{eq:incoherence})
and~(\ref{eq:strong_incoherence}), then we have $\mu_{1}\ge\mu_{0}$
as can be seen from the relations $\sum_{i}\left(\U\V^{\top}\right)_{ij}^{2}=\left\Vert \V^{\top}\ej\right\Vert _{2}^{2}$
and $\sum_{j}\left(\U\V^{\top}\right)_{ij}^{2}=\left\Vert \U^{\top}\ei\right\Vert _{2}^{2}$.
Therefore, the joint incoherence parameter $\mu_{1}$ is the dominant
factor in these previous bounds. As will be discussed in Section~\ref{sub:StrIncoh},
while the standard incoherence~(\ref{eq:incoherence}) is a natural
condition, the joint incoherence condition~(\ref{eq:strong_incoherence})
is restrictive and unintuitive. In several important settings, $\mu_{1}$
is as large as $\mu_{0}^{2}r$, so previous results require $O(nr^{2}\textrm{polylog}(n))$
observations even if $\mu_{0}=O(1)$.

In the following main theorem of the paper, we show that the joint
incoherence is not necessary. The theorem only requires the standard
incoherence condition.
\begin{thm}
\label{thm:uniform}Suppose $\M$ satisfies the standard incoherence
condition \eqref{eq:incoherence} with parameter $\mu_{0}$. There
exist universal constants $c_{0},c_{1},c_{2}>0$ such that if 
\[
p\ge c_{0}\frac{\mu_{0}r\log^{2}(n_{1}+n_{2})}{\min\{n_{1},n_{2}\}},
\]
then $\M$ is the unique optimal solution to~\eqref{eq:cvx_opt}
with probability at least $1-c_{1}(n_{1}+n_{2})^{-c_{2}}$ . 
\end{thm}
We provide comments and discussion in the next two sub-sections.

\subsection{Optimality of Theorem~\ref{thm:uniform}}

Candes and Tao~\cite{candes2010NearOptimalMC} prove the following
\emph{lower-bound} on the sample complexity of matrix completion. 
\begin{prop}
[\cite{candes2010NearOptimalMC}, Theorem 1.7]\label{prop:lower_bound}Suppose
$n_{1}=n_{2}=n$ and $\Omega$ is sampled as above. If we do \emph{not}
have the condition 
\[
p\ge\frac{1}{2}\frac{\mu_{0}r\log(2n)}{n},
\]
and the RHS above is less than $1$, then with probability at least
$\frac{1}{4}$, there exist infinitely many pairs of distinct matrices
$\M^{'}\neq\M^{''}$ of rank at most $r$ and obeying the standard
incoherence condition~\eqref{eq:incoherence} with parameter $\mu_{0}$
such that $M{}_{ij}^{'}=M_{ij}^{''}$ for all $(i,j)\in\Omega$. 
\end{prop}
This shows that that $p\gtrsim\mu_{0}r\log(n)/n$ is necessary for
any method to determine $\M$ (even if one knows $r$ and $\mu_{0}$
ahead of time). With an additional $c'\log(n)$ factor, Theorem~\ref{thm:uniform}
matches this lower bound. In particular, it is optimal in terms of
its scaling with the incoherence parameter\emph{ $\mu_{0}$.}

We note that the condition in Proposition~\ref{prop:lower_bound}
is an \emph{information/statistical} lower-bound: when the value of
$p$ is below this bound, there is not enough information in the observed
entries to uniquely determine an rank-$r$, $\mu_{0}$-incoherent
matrix even if one has infinite computational power. In Section~\ref{sec:LS},
we show that in the closely related problem of matrix decomposition,
the incoherence parameters are associated with both \emph{information}
and \emph{computational} lower bounds.

\subsection{Consequences and Comparison with Prior Work\label{sub:StrIncoh}}

The previous best result for exact matrix completion is given in~\cite{Recht2009SimplerMC,gross2009anybasis}.
They show that $\M\in\mathbb{R}^{n\times n}$ can be recovered by
the nuclear minimization approach if the sampling probabilities satisfy
\[
p\gtrsim\frac{\max\left\{ \mu_{0},\mu_{1}\right\} r\log^{2}n}{n}.
\]
Using an alternative algorithm, Keshavan \textit{et al}.~\cite{keshavan2009matrixafew}
show that recovery can be achieved with 
\[
p\gtrsim\frac{\max\left\{ \mu_{0}r\log n,\mu_{1}^{2}r^{2}\right\} }{n}
\]
Similar results are given in~\cite{jain2013altMin}, which also requires
the sample complexity to be proportional to $\mu_{1}$ (or equivalently,
quadratic in $r$). In light of Proposition~\ref{prop:lower_bound},
these results are not optimal with respect to the incoherence parameters
due to the dependence on the joint incoherence $\mu_{1}$. Theorem~\ref{thm:uniform}
eliminates this extra dependence.

The improvement in Theorem~\ref{thm:uniform} is significant both
qualitatively and quantitatively. The standard incoherence condition
(\ref{eq:incoherence}) is natural and necessary. A small standard
incoherence parameter $\mu_{0}$ ensures that the information of the
row and column spaces of $\M$ is not too concentrated on a small
number of rows/columns. In contrast, the joint incoherence assumption
(\ref{eq:strong_incoherence}), which requires the matrices $\U$
and $\V$ containing the left and right singular vectors to be ``unaligned''
with each other, does not have a natural explanation. In applications,
the quantity $\mu_{0}$ often has clear physical meanings while $\mu_{1}$
does not. For example, in the application to recovering the affinity
matrices between clustered objects from partial observations~\cite{wu2013rating,yi2013assisted}
(discussed in Section~\ref{sec:extension}), $\mu_{0}$ is a function
of the minimum cluster size, but a bound on $\mu_{1}$ bears no natural
motivation. As another example, the work in \cite{chen2013robustSpectralMCicml}
uses Hankel matrix completion to recover spectrally sparse
signals obeying two types of conditions. The first condition can be
traced to standard incoherence and is equivalent to (the natural requirement
of) the supporting frequencies being spread out. On the other hand,
the second set of conditions, which resemble joint incoherence, cannot
be reduced to a property of only the frequencies. The manuscript~\cite{chen2013robustSpectralMC}, which appeared after this paper was posted online~\cite{chen2013incoherence_arxiv}, removes these second set of conditions using similar techniques as in Theorem~\ref{thm:uniform}.

Quantitatively, the joint incoherence condition is much more restrictive
than standard incoherence. By Cauchy-Schwarz inequality we always
have $\mu_{1}r\le\mu_{0}^{2}r^{2}$. The equality 
\[
\mu_{1}r=\mu_{0}^{2}r^{2}
\]
holds in the important setting where the matrix $\M\in\mathbb{R}^{n\times n}$
is positive semidefinite (psd) and thus has $\U=\V$. In this case,
applying previous guarantees would require $p\gtrsim\frac{\mu_{0}^{2}r^{2}\log^{2}n}{n}$.
This translates to $\mu_{0}r$ times more observations than guaranteed
by Theorem~\ref{thm:uniform}. In particular, these previous bounds
are quadratic in $r$, and thus they require $r=o(\sqrt{n})$ regardless
of $p$. This is clearly unnecessary, since any matrix can be completed
regardless of its rank given a sufficiently large $p$. We verify
this fact by simulation. We construct $\M$ as a $0-1$ block diagonal
matrix with $r$ diagonal blocks of size $\frac{n}{r}\times\frac{n}{r}$.
It is easy to see that $\M$ is positive semidefinite with $\mu_{0}=1$
and $\mu_{1}=r$. Figure~\ref{fig:linear} shows the minimum values
of $p$ needed to recover $\M$ for different $r$ in the simulation.
It can be seen that $p$ indeed scales linearly in $r$ as predicted
by Theorem~\ref{thm:uniform}. In particular, we recover PSD matrices
with rank well over $\sqrt{n}$, which would not be possible if the
joint incoherence condition were necessary.

\begin{figure}
\begin{centering}
\includegraphics[scale=0.7]{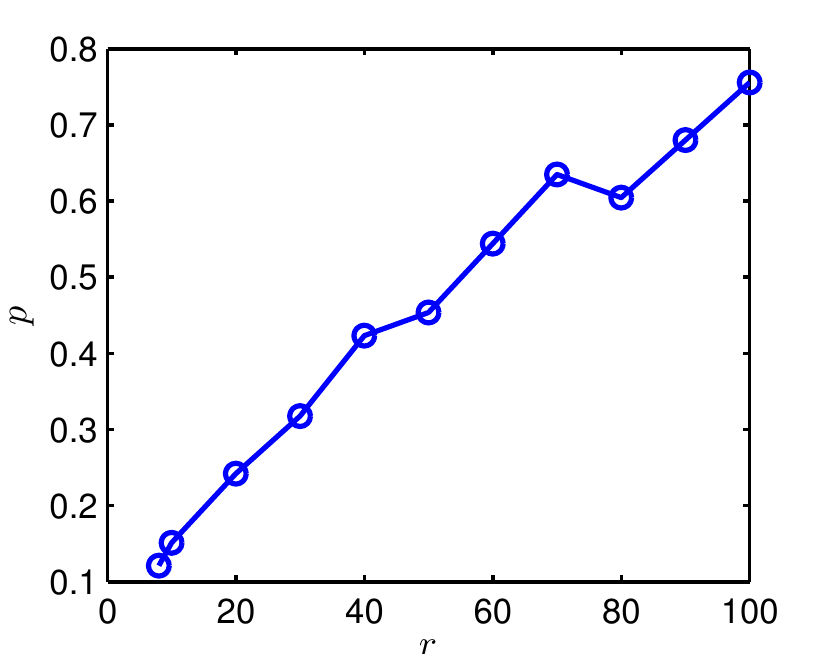} 
\par\end{centering}

\caption{{\small{{\label{fig:linear}The minimum observation probability $p$
for recovering a $900\times900$ rank-$r$ matrix with $\mu_{0}=1$
and $\mu_{1}=r$. We use the IALM method in~\cite{Lin2009_ALM} to
solve the nuclear minimization problem~(\ref{eq:cvx_opt}). For each
$r$ and $p$, we run the simulation for $20$ trials. The $Y$-axis
shows the smallest value of $p$ for which the normalized recovery
error $\left\Vert \hat{\M}-\M\right\Vert _{F}/\left\Vert \M\right\Vert _{F}$
is smaller than $10^{-4}$ in at least $19$ trials. }}}}
\end{figure}

\section{Extensions\label{sec:extension}}

As we mention in the introduction, the proof of Theorem~\ref{thm:uniform}
crucially relies on the use of the matrix $\ell_{\infty,2}$-norm.
In this section, we present two extensions of this idea to the analysis
of an SVD-projection algorithm, and to structured matrix completion
and semi-supervised clustering.

\subsection{Error Bound for SVD Projection}
\label{sec:SVD_proj}

Our first example is the derivation of error bounds for an SVD-projection
algorithm for matrix completion. Let $\P_{\Omega}\M$ be the matrix
obtained from $\M$ by setting all the unobserved entries to zero.
Given the partial observation $\P_{\Omega}\M$, Keshavan \textit{et
al.} \cite{keshavan2009matrixafew} propose the following two-step
algorithm for approximating $\M$. Step 1: Set to zero all columns
and rows in $\P_{\Omega}\M$ with degrees larger that $2pn$, where
the degree of a column or row is the number of non-zero entries of
this column/row. Let $\widetilde{\M}^{\Omega}$ be the output. Step
2: Compute the SVD of $\widetilde{\M}^{\Omega}$ 
\[
\widetilde{\M}^{\Omega}=\sum_{i=1}^{n}\tilde{\sigma}_{i}\tilde{\u}_{i}\tilde{\v}_{j}^{\top}
\]
and return the re-scaled rank-$r$ projection $\textsf{T}_{r}\left(\widetilde{\M}^{\Omega}\right):=\frac{1}{p}\sum_{i=1}^{r}\tilde{\sigma}_{i}\tilde{\u}_{i}\tilde{\v}_{j}$.
Theorem 1.1 in~\cite{keshavan2009matrixafew} provides the following
bound on the approximation error 
\begin{equation}
\left\Vert \M-\textsf{T}_{r}(\widetilde{\M}^{\Omega})\right\Vert _{F}\le c\sqrt{\frac{rn}{p}}\left\Vert \M\right\Vert _{\infty},\quad\mbox{w.h.p.},\label{eq:KMO}
\end{equation}
where $\left\Vert \M\right\Vert _{\infty}:=\max_{i,j}\left|\M_{ij}\right|$
is the matrix $ \ell_\infty $ norm.
This result is proved using a combination of tools from measure
concentration and random graph theory.

As a simple corollary of our Lemma~\ref{lem:op_inf} in Section~\ref{sec:Proofs},
we obtain a new error bound stated in terms of $\left\Vert \M\right\Vert _{\infty}$
and $\left\Vert \M\right\Vert _{\infty,2}$, where $\left\Vert \M\right\Vert _{\infty,2}$ is the matrix $\ell_{\infty,2}$
norm of $ \M $, defined as 
\begin{equation}
\label{eq:Linf2}
\left\Vert \Z\right\Vert _{\infty,2}:=\max\left\{ \max_{i}\sqrt{\sum_{b}\Z_{ib}^{2}},\; \max_{j}\sqrt{\sum_{a}\Z_{aj}^{2}}\right\};
\end{equation}
that is, $\left\Vert \M\right\Vert _{\infty,2}$ is the maximum of the
row and column norms of $\M$. 
\begin{cor}
\label{cor:SVDproj}Suppose $p\ge c_{0}\frac{\log n}{n}$. With high
probability, we have 
\[
\left\Vert \M-\textsf{T}_{r}(\widetilde{\M}^{\Omega})\right\Vert _{F}\le c'\left(\frac{\sqrt{r}\log n}{p}\left\Vert \M\right\Vert _{\infty}+\sqrt{\frac{r\log n}{p}}\left\Vert \M\right\Vert _{\infty,2}\right).
\]
\end{cor}
We prove the corollary in Appendix~\ref{sec:proof_SVDproj}. The corollary improves upon~(\ref{eq:KMO}) whenever $p\gtrsim\log^{2}n/n$
and $\left\Vert \M\right\Vert _{\infty,2}<\sqrt{\frac{n}{\log n}}\left\Vert \M\right\Vert _{\infty}$.
Note that for a general matrix $\M$, $\left\Vert \M\right\Vert _{\infty,2}$
is always no more than $\sqrt{n}\left\Vert \M\right\Vert _{\infty}$,
and can be much smaller. An example of such a matrix is an affinity
matrix with a block-diagonal structure, which is discussed in the
next sub-section. Here the improvement is again due to using the $\ell_{\infty,2}$ norm.

\subsection{Structured Matrix Completion and Semi-Supervised Clustering}

We next consider the extension to the \emph{structured matrix completion}
problem. In several applications of matrix completion including semi-supervised
clustering~\cite{yi2013assisted} (which we shall discuss in more
details) and multi-label learning~\cite{xu2013sideinfo}, one has
access to additional side information about the column/row spaces
of the low-rank matrix $\M$. In particular, one aims to complete
an unknown rank-$r$ matrix $\M=\U\mathbf{\Sigma}\V^{\top}\in\mathbb{R}^{n\times n}$
from the partial observations $\P_{\Omega}\M$, but is given the structural
information that the column (row, resp.) space of $\M$ lie in a known
$\bar{r}$-dimensional subspace of $\mathbb{R}^{n}$ spanned by the
columns of $\bar{\U}\in\mathbb{R}^{n\times\bar{r}}$ ($\bar{\V}\in\mathbb{R}^{n\times\bar{r}}$,
resp.); here $\bar{r}$ may be smaller than the ambient dimension
$n$. In other words, we know $\text{col}(\U)\subseteq\text{col}(\bar{\U})$
and $\text{col}(\V)\subseteq\text{col}(\bar{\V})$, where $\text{col}(\cdot)$
denotes the column space. Without loss of generality, we may assume
each of $\bar{\U}$ and $\bar{\V}$ has orthogonal columns with unit
norms. 

Given $\P_{\Omega}\M$, $\bar{\U}$ and $\bar{\V}$, we solve the
following modified nuclear norm minimization problem: 
\begin{equation}
\begin{aligned}\min_{\X} & \;\left\Vert \X\right\Vert _{*}\\
\textrm{s.t.} & \;\P_{\Omega}(\bar{\U}\X\bar{\V}^{\top})=\P_{\Omega}\M.
\end{aligned}
\label{eq:cluster_prog}
\end{equation}
For this formulation we have the following guarantee.
\begin{thm}
\label{thm:cluster}Suppose $\U$ and $\V$ satisfy the standard incoherence
condition~(\ref{eq:incoherence}) with parameter $\mu_{0}$, and
$\bar{\U}$ and $\bar{\V}$ satisfies$ $~(\ref{eq:incoherence})
with parameter $\bar{\mu}_{0}$. For some universal constants $c_{0},c_{1}$
and $c_{2}$, $\X^{*}:=\bar{\U}^{\top}\M\bar{\V}$ is the unique optimal
solution to the program~(\ref{eq:cluster_prog}) with probability
at least $1-c_{1}n^{-c_{2}}$ provided 
\[
p\ge c_{0}\frac{\mu_{0}\bar{\mu}_{0}r\bar{r}\log(32\bar{\mu}_{0}\bar{r})\log n}{n^{2}}.
\]

\end{thm}
Given $\X^{*}$, we can recover $\M$ by $\M=\bar{\U}\X^{*}\bar{\V}^{\top}$
since $\bar{\U}\bar{\U}^{\top}\U=\U$ and $\bar{\V}\bar{\V}^{\top}\V=\V$.
We prove this theorem in Appendix~\ref{sec:proof_cluster}. 

Theorem~\ref{thm:cluster} shows that with the knowledge of the $\bar{r}$-dimensional
subspaces $\text{col}(\bar{\U})$ and $\text{col}(\bar{\V})$, the
number of observations needed to complete $\M$ is on the order of
$pn^{2}\asymp\mu_{0}\bar{\mu}_{0}r\bar{r}\log(\bar{\mu}_{0}\bar{r})\log(\bar{\mu}_{0}\bar{r})\log n$,
which is $\Theta(r\bar{r}\log\bar{r}\log n)$ for constant $\mu_{0}$
and $\bar{\mu}_{0}$. If $\bar{r}\ll n$, meaning that we have strong
structural information, then this number is much smaller than the
usual requirement $\Theta(nr\log^{2}n)$. On the other hand, setting
$\bar{r}=n$ recovers Theorem~\ref{thm:uniform} for standard matrix
completion where there is no additional structural information. We
note that we assume $\M$ is a square matrix here for simplicity;
the results can be trivially extended to general rectangular matrices.

Near the completion of the writing of this paper,
an independent study~\cite{xu2013sideinfo} on structured matrix
completion was made available. There they require among other things
the following condition:%
\footnote{In \cite{xu2013sideinfo} they consider the sampling without replacement
model for the observed entries. Their results can be translated to
the Bernoulli model considered in this paper, as we have done here.
See also foot note \ref{fn:sampling_model}.%
}
\begin{equation}
p\ge c_{0}\frac{\max\{\mu_{0},\bar{\mu}_{0}\}\max\left\{ \mu_{0},\bar{\mu}_{0},\mu_{1}\right\} r\bar{r}\log\bar{r}\log n}{n^{2}},\label{eq:old_structure_guarantee}
\end{equation}
where $\mu_{1}$ is the joint incoherence parameter of $\U$ and $\V$
defined in~\eqref{eq:strong_incoherence}. Theorem~\ref{thm:cluster}
is better than the result in~\eqref{eq:old_structure_guarantee}
in two ways. First, Theorem~\ref{thm:cluster} avoids the superfluous
dependence on the joint incoherence parameter $\mu_{1}$, which can
be as large as $\mu_{0}^{2}r$ as previously discussed. Second, even
in the ideal setting with $\mu_{1}=\mu_{0}$, the bound in \eqref{eq:old_structure_guarantee}
requires $p$ to scale with $\max\left(\mu_{0}^{2},\bar{\mu}_{0}^{2}\right)$,
whereas the bound in Theorem~\eqref{thm:cluster} scales with $\mu_{0}\bar{\mu}_{0}$,
which is strictly smaller whenever $\mu_{0}\neq\bar{\mu}_{0}$. We
note that \cite{xu2013sideinfo} discusses a nice application of structured
matrix completion to the multi-label learning problem.

\subsubsection{Applications to semi-supervised clustering}

Another interesting application of structured matrix completion is
presented in~\cite{yi2013assisted}. There they consider the semi-supervised
clustering problem, where the goal is to partition a set of $n$ objects
into $r$ clusters given the objects' feature vectors $\w_{i}\in\mathbb{R}^{d},i=1,\ldots,n$
and some must-link and cannot-link constraints $\mathcal{M}$ and
$\mathcal{C}$. In particular, $(i,j)\in\mathcal{M}$ means objects
$i$ and $j$ must be assigned to the same cluster, and $(i,j)\in\mathcal{C}$
means they cannot. Let $\M\in\{0,1\}^{n\times n}$ be the true affinity
matrix, with $M_{ij}=1$ if and only if objects $i$ and $j$ are
in the same cluster. Note that $\M$ has a block-diagonal structure,
with the number of blocks equal to the number of clusters~$r$ and
thus $\text{rank}(\M)=r$. Suppose the rank-$r$ SVD of $\M$ is $\M=\U\mathbf{\Sigma}\U^{\top}$.
The authors of~\cite{yi2013assisted} make the important observation
that in practice, the columns of $\U$ often (approximately) lie in
the space spanned by first $\bar{r}$ singular vectors $\bar{\U}\in\mathbb{R}^{n\times\bar{r}}$
of the input features $\W=[\w_{1}\;\w_{2}\cdots\w_{n}]^{\top}\in\mathbb{R}^{n\times d}$
for some $\bar{r}\ll n$. In this case, one can use the extra information
from the features $\W$ to improve clustering performance. The task
of recovering the affinity matrix $\M$ given $\bar{\U}$ and the
must-link/cannot-link constraints $\Omega:=\mathcal{M}\cup\mathcal{C}$
is precisely a structured matrix completion problem. 

Specifically, \cite{yi2013assisted} considers the setup where the
set of observed entries $\Omega$ are distributed according to the
Bernoulli model with probability $p$,%
\footnote{To be precise, the diagonal entries $M_{ii}=1$ are known; clearly
having more observations cannot decrease the probability that the
program~\eqref{eq:cluster_prog} outputs the correct solution. Moreover,
since the affinity matrix satisfies $M_{ij}=M_{ji}$, each observation
is a \emph{pair} of entries of $\M$. This technicality can be easily
handled, and we omit the details here.%
} the smallest cluster size is $n_{\min}$, and $\bar{\U}$ has standard
incoherence parameter $\bar{\mu}_{0}$ as defined in~(\ref{eq:incoherence}).
Note that the standard incoherence parameter of $\U$ is $n/(rn_{\min})$
due to the block diagonal structure of the affinity matrix $\M$.
Using previous techniques in matrix completion, it is shown in~\cite{yi2013assisted}
that $\X^{*}:=\bar{\U}^{\top}\M\bar{\U}$ is the unique optimal solution
to the program~(\ref{eq:cluster_prog}) w.h.p. provided 
\begin{equation}
p\gtrsim\frac{\bar{\mu}_{0}\bar{r}\log^{2}n}{n_{\min}^{2}}.\label{eq:old_cluster_guarantee}
\end{equation}
Note the quadratic term $n_{\min}^{2}$ on the RHS, which is due to
the joint incoherence parameter of $\U$ taking the value $n^{2}/(rn_{\min}^{2})$.
Suppose $\bar{r}=n$; a consequence of \eqref{eq:old_structure_guarantee}
is that, even if $\M$ is fully observed ($p=1$), the cluster size
must be at least $n_{\min}=\Theta(\sqrt{n})$ and thus the possible
number of clusters $r$ cannot exceed $n/n_{\min}=\Theta(\sqrt{n})$.
These restrictions are undesirable, and clearly unnecessary when $p=1$.

Using Theorem~\ref{thm:cluster}, we can eliminate these $\sqrt{n}$
restrictions and significantly reduce the sample complexity. Plugging
$\mu_{0}=n/(rn_{\min})$ into the theorem, we obtain that the program~(\ref{eq:cluster_prog})
succeeds with high probability provided 
\[
p\gtrsim\frac{\bar{\mu}_{0}\bar{r}\log(\bar{\mu}_{0}\bar{r})\log n}{nn_{\min}}.
\]
The last RHS is order-wise smaller than the RHS of the previous bound~\eqref{eq:old_cluster_guarantee}
by a multiplicative factor of $\frac{n_{\min}}{n}\cdot\frac{\log(\bar{\mu}_{0}\bar{r})}{\log n}$.
In particular, when $\bar{r}=n$ and ignoring logarithm factors, we
allow the size of the clusters to be as small as $n_{\min}=\Theta\left(1\right)$
and the number of clusters be as large as $r=\Theta\left(n\right)$.
These significantly improve over the results in~\cite{yi2013assisted}
which require $n_{\min}=\Omega(\sqrt{n})$ and $r=O(\sqrt{n})$.$ $
Moreover, if $n_{\min}=\sqrt{n}$, then our result require $n/n_{\min}=\sqrt{n}$
times fewer observations than the previous bound~\pageref{eq:old_cluster_guarantee}.

\section{Incoherence in Matrix Decomposition: Information and Computational
Lower Bounds\label{sec:LS}}

Having shown that the joint incoherence is not needed in matrix completion,
we now turn to a closely related problem, namely low-rank and sparse
matrix decomposition~\cite{chandrasekaran2011siam,candes2009robustPCA}.
In contrast to matrix completion, we show that the joint incoherence
condition is unavoidable in matrix decomposition, at least if one
asks for polynomial-time algorithms.

Suppose $\L^{*}\in\mathbb{R}^{n\times n}$ is a symmetric rank-$r$
matrix obeying the standard and joint incoherence conditions~(\ref{eq:incoherence})
and~(\ref{eq:strong_incoherence}) with parameters $\mu_{0}$ and
$\mu_{1}$, respectively, and $\S^{*}\in\mathbb{R}^{n\times n}$ is
a symmetric matrix where each pair of entries $S_{ij}^{*}=S_{ji}^{*}$
is non-zero with probability $\tau$, independent of all others. The
\emph{matrix decomposition} problem concerns with recovering $\left(\L^{*},\S^{*}\right)$
given the sum $\A=\L^{*}+\S^{*}$. A now standard approach is to solve
the following convex program \cite{chandrasekaran2011siam,candes2009robustPCA}:
\begin{equation}
\begin{aligned}\min_{\L,\S} & \;\left\Vert \L\right\Vert _{*}+\lambda\left\Vert \S\right\Vert _{1}\\
\textrm{s.t.} & \;\L+\S=\A,
\end{aligned}
\label{eq:LS}
\end{equation}
where $\left\Vert \S\right\Vert _{1}:=\sum_{i,j}\left|S_{ij}\right|$
is the matrix $\ell_{1}$ norm. Under the above setting, it has been
shown in~\cite{candes2009robustPCA,li2013constantCorruption,chen2011LSarxiv}
that $\left(\L^{*},\S^{*}\right)$ is the unique optimal solution
to~(\ref{eq:LS}) for a suitable $\lambda$ with probability at least
$1-n^{-10}$ provided that $\tau<c_{0}$ for any constant $c_{0}<\frac{1}{2}$
and 
\begin{equation}
c_{1}\frac{\max\left\{ \mu_{0},\mu_{1}\right\} r\log^{2}n}{n}\le1\label{eq:LS_achievability}
\end{equation}
for some constant $c_{1}$ that might depend on $c_{0}$; cf.\ Theorems
1 and 2 in~\cite{chen2011LSarxiv}.%
\footnote{In~\cite{candes2009robustPCA,li2013constantCorruption,chen2011LSarxiv},
the sufficiency of~\eqref{eq:LS_achievability} is proved for non-symmetric
matrices, but it is straightforward to show that the same holds for
the symmetric case considered here.%
} Note the dependence on $\mu_{1}$ above. Consequently, when $\L^{*}$
is positive semi-definite with $\mu_{1}=\mu_{0}^{2}r$, the condition~(\ref{eq:LS_achievability})
requires $r=o(\sqrt{n})$. Unlike the matrix completion setting which
does not have a natural motivation for the $\ell_{\infty}$-type requirement
in the joint incoherence condition~(\ref{eq:strong_incoherence}),
the $\ell_{\infty}$ norm arises naturally in the matrix decomposition
problem as it is the dual norm of the $\ell_{1}$-norm in the formulation~(\ref{eq:LS}).

In fact, we show that the joint incoherence condition is not specific
to the formulation (\ref{eq:LS}), but is in fact required by all
polynomial-time algorithms under a widely-believed computational complexity
assumption. We prove this by connecting the matrix decomposition problem
to the \emph{Planted Clique} problem~\cite{alon1998hiddenClique},
defined as follows. A graph on $n$ nodes is generated by connecting
each pairs of nodes independently with probability~$\frac{1}{2}$,
and then randomly picking a subset of $n_{\min}$ nodes and making
them fully connected (hence a clique). The goal is to find the planted
clique given the graph. The Planted Clique problem has been extensively
studied; cf.~\cite{berthet2013colt_sparsePCA,alon2007testing} for
an overview of the known results. In the regime of $n_{\min}=o(\sqrt{n})$,
there is no known polynomial-time algorithm for this problem despite
years of effort. In fact, this regime is widely believed to be intractable
in polynomial time. The average case hardness of this regime has been
proved under certain computational models~\cite{rossman2010clique,Feldman2012statAlg},
and has been utilized in cryptography~\cite{applebaum2010cryptographyClique,Juel00cliqueCrypto}
and other applications~\cite{alon2007testing,Hazan2011Nash,koiran2012rip}.
The work~\cite{berthet2013colt_sparsePCA} is the first to use this
hardness assumption to obtain bounds on statistical accuracy of sparse
PCA given computational constraints, and a similar approach is taken
in~\cite{ma2013submatrix} for submatrix detection. We therefore
adopt the following computational assumption on the Planted Clique
problem, where we recall that a size $ n_{\min} $ clique is planted in an Erdos-Renyi random graph $ G(n,\frac{1}{2}) $ with $ n $ nodes and edge probability $ \frac{1}{2} $. 
\begin{quote}
\textbf{A1}$\quad$ For any constant $\epsilon>0$, there is no algorithm
with running time polynomial in $n$ that, for all $n$ and with probability at least~$\frac{1}{2}$, finds the
planted clique with size $n_{\min}\le n^{\frac{1}{2}-\epsilon}$ given
the random graph. 
\end{quote}
This version of the assumption is similar to Conjecture 4.3 in~\cite{alon2007testing}.

The following theorem provides necessary conditions for the success
of matrix decomposition algorithms. The proof is given in Appendix~\ref{sec:proof_LS}.
\begin{thm}
\label{thm:LS}The following two statements are true for the matrix
decomposition problem with $\tau=1/3$. 
\begin{enumerate}
\item Suppose $r=1$ and the assumption \textbf{\emph{A1}} holds. For any
constant $\epsilon'>0$, there is no algorithm with running time polynomial
in $n$ that, for all $n$ and with probability at least $\frac{1}{2}$,
solves the matrix decomposition problem%
\footnote{This statement still holds if we restrict to matrix decomposition
problems with $\L^{*}$ and $\S^{*}$ taking finitely many values, which can
be encoded using a finite number of bits. This can be easily seen
from the proof of the theorem.%
} with 
\[
\frac{\mu_{1}^{1-\epsilon'}}{n}\ge1.
\]

\item Suppose $\mu_{0}\ge 2$. There is no algorithm that, for all $n$ and with
probability at least $\frac{1}{2}$, solves the matrix decomposition
problem with 
\[
\frac{1}{12} \cdot \frac{\mu_{0}r \log n }{n}\ge1.
\]

\end{enumerate}
\end{thm}
If we modify the assumption \textbf{A1 }by assuming that the \emph{Planted
$r$-Clique} problem\emph{~}\cite{mcsherry2001spectralpartitioning}
with $r$ disjoint planted cliques of size $o(\sqrt{n})$ is intractable
in polynomial time, then the first part of the theorem holds with
\[
\frac{\mu_{1}^{1-\epsilon'}r}{n}\ge1.
\]
Together with the second part of the theorem, this result shows that,
under the planted clique assumption, the standard and joint incoherence
conditions are both necessary for solving matrix decomposition in
polynomial time. Therefore, the bound in~(\ref{eq:LS_achievability})
is unlikely to be improvable (up to a polylog factor) using polynomial-time
algorithms. In particular, this implies that the matrix decomposition
problem is intractable in general for positive semidefinite matrices
with rank $r=\omega(\sqrt{n})$ since in this case $\mu_{1}r=\mu_{0}^{2}r^{2}\ge r^{2}.$

We note that the first part of Theorem~\ref{thm:LS} is a \emph{computational}
limit. It is proved by showing that if there is a matrix decomposition
algorithm that does not require the joint incoherence condition, then
the algorithm would solve the computationally hard problem of finding
a planted clique with size $n_{\min}=o(\sqrt{n})$. On the other hand,
the second part of the theorem is an \emph{information/statistical}
limit applicable to all algorithms regardless of their computational
complexity, and is proved by an information-theoretic argument. Interestingly,
Theorem~\ref{thm:LS} shows that the standard incoherence and the
joint incoherence are associated with the statistical and computational
aspects of the matrix decomposition problem, respectively.

\section{Proof of Theorem~\ref{thm:uniform} \label{sec:Proofs}}

We prove the our main result Theorem~\ref{thm:uniform} in this section.
While Theorem~\ref{thm:uniform} can be derived from the more general
Theorem~\ref{thm:cluster}, we choose to provide a separate proof
of Theorem~\ref{thm:uniform} in order to highlight the main innovation
(the use of the $\ell_{\infty,2}$ norm) of the analysis. The general
setting of Theorem~\ref{thm:cluster} requires several additional
technical steps.

The high level roadmap of the proof is a standard one: by convex analysis,
to show that $\M$ is the unique optimal solution to the program~(\ref{eq:cvx_opt}),
it suffices to construct a \emph{dual certificate }$\Y$ obeying several
subgradient-type conditions. One of the conditions requires the spectral
norm $\left\Vert \Y\right\Vert $ to be small. Previous work bounds
$\left\Vert \Y\right\Vert $ by the $\ell_{\infty}$ norm $\left\Vert \Z\right\Vert _{\infty}:=\max_{i,j}\left|\Z_{ij}\right|$
of a certain matrix $\Z$, which ultimately links to $\left\Vert \U\V^{\top}\right\Vert _{\infty}$
and thus leads to the joint incoherence condition in~(\ref{eq:strong_incoherence}).
Here we derive a new bound using the $\ell_{\infty,2}$ norm $\left\Vert \Z\right\Vert _{\infty,2}$ as defined in~\eqref{eq:Linf2}.
Note that $\left\Vert \Z\right\Vert _{\infty,2}$ is always no greater
than $\sqrt{\max\{n_{1},n_{2}\}}\left\Vert \Z\right\Vert _{\infty}$
for any $\Z\in\mathbb{R}^{n_{1}\times n_{2}}$. In our setting, there
is a significant gap between $\left\Vert \U\V^\top \right\Vert _{\infty,2}\le\sqrt{\frac{\mu_{0}r}{\min\{n_{1},n_{2}\}}}$
and $\sqrt{\max\{n_{1},n_{2}\}}\left\Vert \U\V^\top\right\Vert _{\infty}\le\frac{\mu_{0}r}{\sqrt{\min\{n_{1},n_{2}\}}}.$
This leads to a tighter bound of $\left\Vert \Y\right\Vert $ and
hence less restrictive incoherence conditions.

We now turn to the details. To simplify notion, we prove the results
for square matrices~($n_{1}=n_{2}=n$); the results for non-square
matrices are proven in exactly the same fashion. Some additional notation
is needed. We use $c$ and its derivatives ($c',c_{0}$, etc.) for
universal positive constants.  By \emph{with high probability} (\emph{w.h.p.})
we mean with probability at least $1-c_{1}n{}^{-c_{2}}$ for some
constants $c_{1},c_{2}>0$ independent of the problem parameters ($n,r,p,\mu_{0},\mu_{1}$).
Throughout the proof the constant $c_{2}$ can be made arbitrarily
large by choosing the constant $c_{0}$ in Theorem~\ref{thm:uniform}
sufficiently large. The proof below involves $80\log n + 1$ random events, each of which is shown to occur with
high probability. By the union bound their intersection also
occurs with high probability. 

A few additional notations are needed. The inner product between two matrices
is given by $\left\langle \X,\Z\right\rangle :=\mbox{trace}(\X^{\top}\Z)$.
The projections $\P_{T}$ and $\P_{T^{\bot}}$ are given by 
\[
\P_{T}(\Z):=\U\U^{\top}\Z+\Z\V\V^{\top}-\U\U^{\top}\Z\V\V^{\top}
\]
and $\P_{T^{\bot}}(\Z):=\Z-\P_{T}(\Z).$ $\P_{\Omega}(\Z)$ denotes
the matrix given by $\left(\P_{\Omega}(\Z)\right)_{ij}=Z_{ij}$ if $(i,j)\in\Omega$
and zero otherwise. We use $\II$ to denote the identity mapping for
matrices. For $1\le i,j\le n$, we define the random variable $\gamma_{ij}:=\mathbb{I}\left((i,j)\in\Omega\right)$,
where $\mathbb{I}(\cdot)$ is the indicator function. The projection
$\R_{\Omega}$ is given by 
\begin{equation}
\R_{\Omega}\Z:=\frac{1}{p}\P_{\Omega}\Z=\sum_{i,j}\frac{1}{p}\gamma_{ij}Z_{ij}\ei\ej^{\top}.
\end{equation}
As usual, $\left\Vert \z\right\Vert _{2}$ is the $\ell_{2}$ norm
of the vector $\z$, and $\left\Vert \Z\right\Vert _{F}$ and $\left\Vert \Z\right\Vert $
are the Frobenius norm and spectral norm of the matrix $\Z$, respectively.
For an operator $\mathcal{A}$ on matrices, its operator norm is defined
as $\left\Vert \mathcal{A}\right\Vert _{op}:=\sup_{\Z\in\mathbb{R}^{n\times n}}\left\Vert \mathcal{A}(\Z)\right\Vert _{F}/\left\Vert \Z\right\Vert _{F}.$

\paragraph*{Subgradient Optimality Condition}

Following our proof roadmap, we now state a sufficient condition for
$\M$ to be the unique optimal solution to the optimization problem
(\ref{eq:cvx_opt}). 
\begin{prop}
\label{prop:opt_cond} Suppose $p\ge\frac{1}{n}$. The matrix $\M$
is the unique optimal solution to~(\ref{eq:cvx_opt}) if the following
conditions hold: 
\begin{enumerate}
\item $\left\Vert \P_{T}\R_{\Omega}\P_{T}-\P_{T}\right\Vert _{op}\le\frac{1}{2}.$ 
\item There exists a dual certificate $\Y\in\mathbb{R}^{n\times n}$ which
satisfies $\P_{\Omega}(\Y)=\Y$ and

\begin{enumerate}
\item $\left\Vert \P_{T^{\bot}}(\Y)\right\Vert \le\frac{1}{2}$,
\item $\left\Vert \P_{T}(\Y)-\U\V^{\top}\right\Vert _{F}\le\frac{1}{4n}.$ 
\end{enumerate}
\end{enumerate}
\end{prop}
A somewhat different version of the proposition appears in~\cite{Recht2009SimplerMC,gross2009anybasis}.
We prove the proposition in Appendix~\ref{sec:proof_opt_cond}.

\paragraph*{Approximate Isometry}

The requirement $p\ge\frac{1}{n}$ in Proposition~\ref{prop:opt_cond}
clearly holds under the conditions of Theorem~\ref{thm:uniform}.
The following standard result shows that the approximate isometry
in Condition 1 is also satisfied. 
\begin{lem}
[Theorem 4.1 in~\cite{candes2009exact}; Lemma 11 in~\cite{chen2011LSarxiv}]\label{lem:op}If
$p\ge c_{0}\frac{\mu_{0}r\log n}{n}$ for some $c_{0}$ sufficiently
large, then w.h.p. 
\[
\left\Vert \P_{T}\R_{\Omega}\P_{T}-\P_{T}\right\Vert \le\frac{1}{2}.
\]

\end{lem}

\paragraph*{Constructing the Dual Certificate}

We now construct a dual certificate $\Y$ that satisfies Condition
2 in Proposition~\ref{prop:opt_cond}. We do this using the Golfing
Scheme~\cite{gross2009anybasis,candes2009robustPCA}. Set $k_{0}:=20\log n$.
Assume for now the set $\Omega$ of observed entries is generated
from $\Omega=\bigcup_{k=1}^{k_{0}}\Omega_{k}$, where for each $k$
and matrix index $(i,j)$, $\mathbb{P}\left[(i,j)\in\Omega_{k}\right]=q:=1-(1-p)^{1/k_{0}}$
and is independent of all others. Clearly this $\Omega$ has the same
distribution as the original model. Let $\W_{0}:=0$ and for $k=1,\ldots,k_{0},$
define 
\begin{equation}
\W_{k}:=\W_{k-1}+\R_{\Omega_{k}}\P_{T}\left(\U\V^{\top}-\P_{T}\W_{k-1}\right),\label{eq:W_defn}
\end{equation}
where the operator $\R_{\Omega_{k}}$ is defined analogously to $\R_{\Omega}$
as $\R_{\Omega_{k}}(\Z):=\sum_{i,j}\frac{1}{q}\mathbb{I}\left((i,j)\in\Omega_{k}\right)Z_{ij}\ei\ej^{\top}.$
The dual certificate is given by $\Y:=\W_{k_{0}}.$ We have $\P_{\Omega}(\Y)=\Y$
by construction. The proof of Theorem~\ref{thm:uniform} is completed
if we show that $\Y$ satisfies Conditions 2(a) and 2(b) in Proposition
\ref{prop:opt_cond}~w.h.p.

\paragraph*{Lemmas on Matrix Norms}

The key step of our proof is to show that $\Y$ satisfies Condition
2(a) in Proposition~\ref{prop:opt_cond}, i.e., we need to bound
$\left\Vert \P_{T^{\bot}}(\Y)\right\Vert $. Here our proof departs
from existing work -- we establish bounds on this quantity in terms
of the $\ell_{\infty,2}$ norm. This is done with the help of two
lemmas. The first one bounds the spectral norm of $\left(\R_{\Omega}-\II\right)\Z$
in terms of the $\ell_{\infty,2}$ and $\ell_{\infty}$ norms of $\Z$.
This gives tighter bounds than previous approaches~\cite{candes2010NearOptimalMC,gross2010noreplacement,Recht2009SimplerMC,keshavan2009matrixafew}
that use solely the $\ell_{\infty}$ norm of $\Z$. 
\begin{lem}
\label{lem:op_inf} Suppose $\Z$ is a fixed $n\times n$ matrix.
For a universal constant $c>1$, we have w.h.p. 
\[
\left\Vert \left(\R_{\Omega}-\II\right)\Z\right\Vert \le c\left(\frac{\log n}{p}\left\Vert \Z\right\Vert _{\infty}+\sqrt{\frac{\log n}{p}}\left\Vert \Z\right\Vert _{\infty,2}\right).
\]

\end{lem}
The second lemma further controls the $\ell_{\infty,2}$ norm.
\begin{lem}
\label{lem:inf_2}Suppose $\Z$ is a fixed matrix. If $p\ge c_{0}\frac{\mu_{0}r\log n}{n}$
for some $c_{0}$ sufficiently large, then w.h.p. 
\[
\left\Vert \left(\P_{T}\R_{\Omega}-\P_{T}\right)\Z\right\Vert _{\infty,2}\le\frac{1}{2}\sqrt{\frac{n}{\mu_{0}r}}\left\Vert \Z\right\Vert _{\infty}+\frac{1}{2}\left\Vert \Z\right\Vert _{\infty,2}.
\]

\end{lem}
We prove Lemmas~\ref{lem:op_inf} and~\ref{lem:inf_2} in Appendix
\ref{sec:proof_tech}. We also need a standard result that controls
the $\ell_{\infty}$ norm. 
\begin{lem}
[Lemma 3.1 in~\cite{candes2009robustPCA}; Lemma 13 in~\cite{chen2011LSarxiv}]\label{lem:inf}Suppose
$\Z$ is a fixed $n\times n$ matrix in $T$. If $p\ge c_{0}\frac{\mu r\log n}{n}$
for some $c_{0}$ sufficiently large, then w.h.p. 
\[
\left\Vert \left(\P_{T}\R_{\Omega}\P_{T}-\P_{T}\right)\Z\right\Vert _{\infty}\le\frac{1}{2}\left\Vert \Z\right\Vert _{\infty}.
\]

\end{lem}
Equipped with the lemmas above, we are ready to validate Condition
2 in Proposition~\ref{prop:opt_cond}.

\paragraph{Validating Condition 2(b)}

Set $\D_{k}:=\U\V^{\top}-\P_{T}(\W_{k})$ for $k=0,\ldots,k_{0}$.
By definition of $\W_{k}$, we have $\D_{0}=\U\V^{\top}$and 
\begin{align}
\D_{k} & =\left(\P_{T}-\P_{T}\R_{\Omega_{k}}\P_{T}\right)\D_{k-1}.\label{eq:recursion}
\end{align}
Note that $\Omega_{k}$ is independent of $\D_{k-1}$ and $q\ge p/k_{0}\ge c_{0}\mu_{0}r\log(n)/n$
under the conditions in Theorem~\ref{thm:uniform}. Applying Lemma
\ref{lem:op} with $\Omega$ replaced by $\Omega_{k}$, we obtain
that w.h.p. 
\[
\left\Vert \D_{k}\right\Vert _{F}\le\left\Vert \P_{T}-\P_{T}\R_{\Omega_{k}}\P_{T}\right\Vert \left\Vert \D_{k-1}\right\Vert _{F}\le\frac{1}{2}\left\Vert \D_{k-1}\right\Vert _{F}
\]
for each $k$. Applying the above inequality recursively with $k=k_{0,}k_{0}-1,\ldots,1$
gives 
\[
\left\Vert \P_{T}(\Y)-\U\V^{\top}\right\Vert _{F}=\left\Vert \D_{k_{0}}\right\Vert _{F}\le\left(\frac{1}{2}\right)^{k_{0}}\left\Vert \U\V^{\top}\right\Vert _{F}\le\frac{1}{4n^{2}}\cdot\sqrt{r}\le\frac{1}{4n}.
\]

\paragraph{Validating Condition 2(a)}

Note that $\Y=\sum_{k=1}^{k_{0}}\R_{\Omega_{k}}\P_{T}\left(\D_{k-1}\right)$
by construction. We therefore have 
\begin{align*}
\left\Vert \P_{T^{\bot}}(\Y)\right\Vert  & \le\sum_{k=1}^{k_{0}}\left\Vert \P_{T^{\bot}}\left(\R_{\Omega_{k}}\P_{T}-\P_{T}\right)\left(\D_{k-1}\right)\right\Vert \le\sum_{k=1}^{k_{0}}\left\Vert \left(\R{}_{\Omega_{k}}-\II\right)\P_{T}\left(\D_{k-1}\right)\right\Vert .
\end{align*}
Applying Lemma~\ref{lem:op_inf} with $\Omega$ replaced by $\Omega_{k}$
to each summand of the last R.H.S., we get that w.h.p. 
\begin{align}
\left\Vert \P_{T^{\bot}}(\Y)\right\Vert  & \le c\sum_{k=1}^{k_{0}}\left(\frac{\log n}{q}\left\Vert \D_{k-1}\right\Vert _{\infty}+\sqrt{\frac{\log n}{q}}\left\Vert \D_{k-1}\right\Vert _{\infty,2}\right)\nonumber \\
 & \le\frac{c}{\sqrt{c_{0}}}\sum_{k=1}^{k_{0}}\left(\frac{n}{\mu_{0}r}\left\Vert \D_{k-1}\right\Vert _{\infty}+\sqrt{\frac{n}{\mu_{0}r}}\left\Vert \D_{k-1}\right\Vert _{\infty,2}\right),\label{eq:x}
\end{align}
where the last inequality follows from $q\ge c_{0}\mu_{0}r\log(n)/n$.
We proceed by bounding $\left\Vert \D_{k-1}\right\Vert _{\infty}$
and $\left\Vert \D_{k-1}\right\Vert _{\infty,2}$. Using~\eqref{eq:recursion},
and repeatedly applying Lemma~\ref{lem:inf} with $\Omega$ replaced
by $\Omega_{k}$, we obtain that w.h.p. 
\begin{align*}
\left\Vert \D_{k-1}\right\Vert _{\infty} & =\left\Vert \left(\P_{T}-\P_{T}\R_{\Omega_{k-1}}\P_{T}\right)\cdots\left(\P_{T}-\P_{T}\R_{\Omega_{1}}\P_{T}\right)\D_{0}\right\Vert _{\infty}\le\left(\frac{1}{2}\right)^{k-1}\left\Vert \U\V^{\top}\right\Vert _{\infty}.
\end{align*}
By Lemma~\ref{lem:inf_2} with $\Omega$ replaced by $\Omega_{k}$,
we obtain that~w.h.p. 
\begin{align*}
\left\Vert \D_{k-1}\right\Vert _{\infty,2} 
= & \left\Vert \left(\P_{T}-\P_{T}\R_{\Omega_{k-1}}\P_{T}\right)\D_{k-2}\right\Vert _{\infty,2}\le\frac{1}{2}\sqrt{\frac{n}{\mu r}}\left\Vert \D_{k-2}\right\Vert _{\infty}+\frac{1}{2}\left\Vert \D_{k-2}\right\Vert _{\infty,2}.
\end{align*}
Using~\eqref{eq:recursion} and combining the last two display equations
gives w.h.p. 
\[
\left\Vert \D_{k-1}\right\Vert _{\infty,2}\le k\left(\frac{1}{2}\right)^{k-1}\sqrt{\frac{n}{\mu r}}\left\Vert \U\V^{\top}\right\Vert _{\infty}+\left(\frac{1}{2}\right)^{k-1}\left\Vert \U\V^{\top}\right\Vert _{\infty,2}.
\]
Substituting back to~\eqref{eq:x}, we get w.h.p. 
\begin{align*}
\left\Vert \P_{T^{\bot}}(\Y)\right\Vert  
& \le \frac{c}{\sqrt{c_{0}}}\frac{n}{\mu_{0}r}\left\Vert \U\V^{\top}\right\Vert _{\infty}\sum_{k=1}^{k_{0}}(k+1)\left(\frac{1}{2}\right)^{k-1}+\frac{c}{\sqrt{c_{0}}}\sqrt{\frac{n}{\mu_{0}r}}\left\Vert \U\V^{\top}\right\Vert _{\infty,2}\sum_{k=1}^{k_{0}}\left(\frac{1}{2}\right)^{k-1}\\
& \le \frac{6c}{\sqrt{c_{0}}}\frac{n}{\mu_{0}r}\left\Vert \U\V^{\top}\right\Vert _{\infty}+\frac{2c}{\sqrt{c_{0}}}\sqrt{\frac{n}{\mu_{0}r}}\left\Vert \U\V^{\top}\right\Vert _{\infty,2}.
\end{align*}
But the standard incoherence condition~\eqref{eq:incoherence} implies
that 
\begin{align*}
\left\Vert \U\V^{\top}\right\Vert _{\infty} & \le\max_{i,j}\left\Vert \U^{\top}\ei\right\Vert _{2}\left\Vert \V^{\top}\ej\right\Vert _{2}\le\frac{\mu_{0}r}{n},\\
\left\Vert \U\V^{\top}\right\Vert _{\infty,2} & \le\max\left\{ \max_{i}\left\Vert \ei^{\top}\U\V^{\top}\right\Vert _{2},\max_{j}\left\Vert \U\V^{\top}\ej\right\Vert _{2}\right\} \le\sqrt{\frac{\mu_{0}r}{n}}.
\end{align*}
It follows that w.h.p. 
\[
\left\Vert \P_{T^{\bot}}(\Y)\right\Vert \le\frac{6c}{\sqrt{c_{0}}}+\frac{2c}{\sqrt{c_{0}}}\le\frac{1}{2}
\]
provided $c_{0}$ is sufficiently large. This completes the proof
of Theorem~\ref{thm:uniform}.

\section{Discussion\label{sec:discussion}}

In this paper, we consider exact matrix completion and show that the
joint incoherence condition imposed by all previous work is in fact
not necessary. We discuss two extensions of this result, namely in
bounding the approximation errors of SVD projection, and in structured
matrix completion and semi-supervised clustering. We then show that
the joint incoherence condition is unavoidable in the apparently similar
problem of low-rank and sparse matrix decomposition based on the computational
hardness assumption of the Planted Clique problem.

The improvements in the matrix completion problem are achieved via
the use of $\ell_{\infty,2}$-type bounds. The $\ell_{\infty,2}$
norm seems to be\emph{ }natural in the context of low-rank matrices
as it captures the relative importance of the rows and columns. It
is interesting to see if the techniques in this paper are relevant
more generally.

\section*{Acknowledgment}

The author would like to thank Constantine Caramanis, Yuxin Chen,
Sujay Sanghavi and Rachel Ward for their support and helpful comments. This work is supported by NSF grant EECS-1056028 and DTRA grant HDTRA 1-08-0029.

\appendixpage

\appendix

\section{Proof of Proposition~\ref{prop:opt_cond}\label{sec:proof_opt_cond}}

Consider any feasible solution $\X$ to~(\ref{eq:cvx_opt}) with
$\P_{\Omega}(\X)=\P_{\Omega}(\M)$. Let $\G$ be an $n\times n$ matrix
which satisfies $\left\Vert \P_{T^{\bot}}\G\right\Vert =1$ and $\left\langle \P_{T^{\bot}}\G,\P_{T^{\bot}}(\X-\M)\right\rangle =\left\Vert \P_{T^{\bot}}(\X-\M)\right\Vert _{*}$.
Such $G$ always exists by duality between the nuclear norm and the
spectral norm. Because $\U\V^{\top}+\P_{T^{\bot}}\G$ is a sub-gradient
of $\left\Vert \Z\right\Vert _{*}$ at $\Z=\M$, we get 
\[
\left\Vert \X\right\Vert _{*}-\left\Vert \M\right\Vert _{*}\ge\left\langle \U\V^{\top}+\P_{T^{\bot}}\G,\X-\M\right\rangle .
\]
We also have $\left\langle \Y,\X-\M\right\rangle =\left\langle \P_{\Omega}(\Y),\P_{\Omega}(\X-\M)\right\rangle =0$
since $\P_{\Omega}(\Y)=\Y$. It follows that 
\begin{align*}
\left\Vert \X\right\Vert _{*}-\left\Vert \M\right\Vert _{*} & \ge\left\langle \U\V^{\top}+\P_{T^{\bot}}\G-\Y,\X-\M\right\rangle \\
 & =\left\Vert \P_{T^{\bot}}(\X-\M)\right\Vert _{*}+\left\langle \U\V^{\top}-\P_{T}\Y,\X-\M\right\rangle -\left\langle \P_{T^{\bot}}\Y,\X-\M\right\rangle \\
 & \ge\left\Vert \P_{T^{\bot}}(\X-\M)\right\Vert _{*}-\left\Vert \U\V^{\top}-\P_{T}\Y\right\Vert _{F}\left\Vert \P_{T}(\X-\M)\right\Vert _{F}-\left\Vert \P_{T^{\bot}}\Y\right\Vert \left\Vert \P_{T^{\bot}}(\X-\M)\right\Vert _{*}\\
 & \ge\frac{1}{2}\left\Vert \P_{T^{\bot}}(\X-\M)\right\Vert _{*}-\frac{1}{4n^{5}}\left\Vert \P_{T}(\X-\M)\right\Vert _{F},
\end{align*}
where in the last inequality we use Conditions 1 and 2 in the proposition.
Applying Lemma~\ref{lem:mini_lemma} below, we further obtain 
\begin{align*}
\left\Vert \X\right\Vert _{*}-\left\Vert \M\right\Vert _{*} & \ge\frac{1}{2}\left\Vert \P_{T^{\bot}}(\X-\M)\right\Vert _{*}-\frac{1}{4n^{5}}\cdot\sqrt{2}n^{5}\left\Vert \P_{T^{\bot}}(\X-\M)\right\Vert _{*}>\frac{1}{8}\left\Vert \P_{T^{\bot}}(\X-\M)\right\Vert _{*}.
\end{align*}
The RHS is strictly positive for all $\X$ with $\P_{\Omega}(\X-\M)=0$
and $\X\neq\M$. Otherwise we must have $\P_{T}(\X-\M)=\X-\M$ and
$\P_{T}\R_{\Omega}\P_{T}(\X-\M)=0$, contradicting the assumption
$\left\Vert \P_{T}\R_{\Omega}\P_{T}-\P_{T}\right\Vert _{op}\le\frac{1}{2}$.
This proves that $\M$ is the unique optimum.
\begin{lem}
\label{lem:mini_lemma}If $p\ge\frac{1}{n^{10}}$ and $\left\Vert \P_{T}\R_{\Omega}\P_{T}-\P_{T}\right\Vert _{op}\le\frac{1}{2}$,
then we have 
\[
\left\Vert \P_{T}\Z\right\Vert _{F}\le\sqrt{2}n^{5}\left\Vert \P_{T^{\bot}}(\Z)\right\Vert _{*},\forall\Z\in\{\Z':\P_{\Omega}(\Z')=0\}.
\]
\end{lem}
\begin{proof}
Observe that 
\begin{align*}
\left\Vert \sqrt{p}\R_{\Omega}\P_{T}(\Z)\right\Vert _{F} & =\sqrt{\left\langle \left(\P_{T}\R_{\Omega}\P_{T}-\P_{T}\right)\Z,\P_{T}(\Z)\right\rangle +\left\langle \P_{T}(\Z),\P_{T}(\Z)\right\rangle }\\
 & \ge\sqrt{\left\Vert \P_{T}(\Z)\right\Vert _{F}^{2}-\left\Vert \P_{T}\R_{\Omega}\P_{T}-\P_{T}\right\Vert \left\Vert \P_{T}(\Z)\right\Vert _{F}^{2}}\ge\frac{1}{\sqrt{2}}\left\Vert \P_{T}(\Z)\right\Vert _{F},
\end{align*}
where the last inequality follows from the assumption $\left\Vert \P_{T}\R_{\Omega}\P_{T}-\P_{T}\right\Vert _{op}\le\frac{1}{2}$.
On the other hand, $\P_{\Omega}(\Z)=0$ implies $\R_{\Omega}(\Z)=0$
and thus 
\[
\left\Vert \sqrt{p}\R_{\Omega}\P_{T}(\Z)\right\Vert _{F}=\left\Vert \sqrt{p}\R_{\Omega}\P_{T^{\bot}}(\Z)\right\Vert _{F}\le\frac{1}{\sqrt{p}}\left\Vert \P_{T^{\bot}}(\Z)\right\Vert _{F}\le n^{5}\left\Vert \P_{T^{\bot}}(\Z)\right\Vert _{F}.
\]
Combining the last two display equations gives 
\[
\left\Vert \P_{T}(\Z)\right\Vert _{F}\le\sqrt{2}n^{5}\left\Vert \P_{T^{\bot}}(\Z)\right\Vert _{F}\le\sqrt{2}n^{5}\left\Vert \P_{T^{\bot}}(\Z)\right\Vert _{*}.
\]

\end{proof}

\section{Proofs of Technical Lemmas\label{sec:proof_tech} in Section~\ref{sec:Proofs}}

We prove the technical lemmas that are used in the proof of Theorem
\ref{thm:uniform}. The proofs use the matrix Bernstein inequality,
restated below. 
\begin{thm}
[\cite{tropp2010matrixmtg}]\label{lem:matrix_bernstein}Let $\X_{1},\ldots,\X_{N}\in\mathbb{R}^{n}$
be independent zero mean random matrices. Suppose 
\[
\max\left\{ \left\Vert \mathbb{E}\sum_{k=1}^{N}\X_{k}\X_{k}^{\top}\right\Vert ,\left\Vert \mathbb{E}\sum_{k=1}^{N}\X_{k}^{\top}\X_{k}\right\Vert \right\} \le\sigma^{2}
\]
and $\left\Vert \X_{k}\right\Vert \le B$ almost surely for all $k$.
Then for any $c>1$, we have 
\[
\left\Vert \sum_{k=1}^{N}\X_{k}\right\Vert \le\sqrt{4c\sigma^{2}\log(2n)}+cB\log(2n).
\]
with probability at least $1-(2n)^{-(c-1)}.$ 
\end{thm}
We also make use of the following facts: for all $i$ and $j$, we
have 
\begin{equation}
\left\Vert \P_{T}(\ei\ej^{\top})\right\Vert _{F}^{2}\le\frac{2\mu_{0}r}{n}.\label{eq:basic_inequality}
\end{equation}
This follows from the definition of $\P_{T}$ and the standard incoherence
condition~\eqref{eq:incoherence}.

\subsection{Proof of Lemma~\ref{lem:op_inf}}

We may write 
\[
\left(\R_{\Omega}-\II\right)\Z=\sum_{i,j}\S_{(ij)}:=\sum_{i,j}\left(\frac{1}{p}\gamma_{ij}-1\right)Z_{ij}\ei\ej^{\top},
\]
where $\left\{ \S_{(ij)}\right\} $ are independent matrices satisfying
$\mathbb{E}[\S_{(ij)}]=0$ and $\left\Vert \S_{(ij)}\right\Vert \le\frac{1}{p}\left\Vert \Z\right\Vert _{\infty}.$
Moreover, we have 
\[
\mathbb{E}\sum_{i,j}\S_{(ij)}^{\top}\S_{(ij)}
=\sum_{i,j}Z_{ij}^{2}\ej\ei^{\top}\ei\ej^{\top}\mathbb{E}\left(\frac{1}{p}\gamma_{ij}-1\right)^{2}
=\sum_{i,j}\frac{1-p}{p}Z_{ij}^{2}\ej\ej^{\top}
\]
and thus 
\[
\left\Vert \mathbb{E}\sum_{i,j}\S_{(ij)}^{\top}\S_{(ij)}\right\Vert \le\frac{1}{p}\max_{j}\left|\sum_{i=1}^{n}Z_{ij}^{2}\right|\le\frac{1}{p}\left\Vert \Z\right\Vert _{\infty,2}^{2}.
\]
In a similar way we can bound $\left\Vert \mathbb{E}\sum_{i,j}\S_{(ij)}\S_{(ij)}^{\top}\right\Vert $
by the same quantity. Applying the matrix Bernstein inequality in
Theorem~\ref{lem:matrix_bernstein} proves the lemma.

\subsection{Proof of Lemma~\ref{lem:inf_2}}

Fix $b\in[n]$. The $b$-th column of the matrix $\left(\P_{T}\R_{\Omega}-\P_{T}\right)\Z$
can be written as 
\[
\left(\left(\P_{T}\R_{\Omega}-\P_{T}\right)\Z\right)\e_{b}=\sum_{i,j}\s_{(ij)}:=\sum_{i,j}\left(\frac{1}{p}\gamma_{ij}-1\right)Z_{ij}\P_{T}(\ei\ej^{\top})\e_{b},
\]
where $\{\s_{(ij)}\}$ are independent column vectors in $\mathbb{R}^{n}$.
Note that $\mathbb{E}\left[\s_{(ij)}\right]=0$ and
\[
\left\Vert \s_{(ij)}\right\Vert _{2}\le\frac{1}{p}\sqrt{\frac{\mu_{0}r}{n}}\left\Vert \Z\right\Vert _{\infty}\le\frac{1}{c_{0}\log n}\sqrt{\frac{n}{\mu_{0}r}}\left\Vert \Z\right\Vert _{\infty},
\]
where the last inequality follows from the assumption of $p$ in the
statement of the lemma. We also have 
\[
\left|\mathbb{E}\left[\sum_{i,j}\s_{(ij)}^{\top}\s_{(ij)}\right]\right|=\left|\sum_{i,j}\mathbb{E}\left[\left(\frac{1}{p}\gamma_{ij}-1\right)\right]Z_{ij}^{2}\left\Vert \P_{T}(\ei\ej^{\top})\e_{b}\right\Vert _{2}^{2}\right|=\frac{1-p}{p}\sum_{i,j}Z_{ij}^{2}\left\Vert \P_{T}(\ei\ej^{\top})\e_{b}\right\Vert _{2}^{2}.
\]
Observe that 
\[
\left\Vert \P_{T}(\ei\ej^{\top})\e_{b}\right\Vert _{2}=\left\Vert \U\U^{\top}\ei\ej^{\top}\e_{b}+(\I-\U\U^{\top})\ei\ej^{\top}\V\V^{\top}\e_{b}\right\Vert _{2}\le\sqrt{\frac{\mu_{0}r}{n}}\left|\ej^{\top}\e_{b}\right|+\left|\ej^{\top}\V\V^{\top}\e_{b}\right|
\]
using the incoherence condition~\eqref{eq:incoherence}. It follows
that 
\begin{align*}
\left|\mathbb{E}\left[\sum_{i,j}\s_{(ij)}^{\top}\s_{(ij)}\right]\right| & \le\frac{2}{p}\sum_{i,j}Z_{ij}^{2}\frac{\mu_{0}r}{n}\left|\ej^{\top}\e_{b}\right|^{2}+\frac{2}{p}\sum_{i,j}Z_{ij}^{2}\left|\ej^{\top}\V\V^{\top}\e_{b}\right|^{2}\\
 & =\frac{2\mu_{0}r}{pn}\sum_{i}Z_{ib}^{2}+\frac{2}{p}\sum_{j}\left|\ej^{\top}\V\V^{\top}\e_{b}\right|^{2}\sum_{i}Z_{ij}^{2}\\
 & \le\frac{2}{p}\frac{\mu_{0}r}{n}\left\Vert \Z\right\Vert _{\infty,2}^{2}+\frac{2}{p}\left\Vert \V\V^{\top}\e_{b}\right\Vert ^{2}\left\Vert \Z\right\Vert _{\infty,2}^{2}\\
 & \le\frac{4\mu_{0}r}{pn}\left\Vert \Z\right\Vert _{\infty,2}^{2}\le\frac{4}{c_{0}\log n}\left\Vert \Z\right\Vert _{\infty,2}^{2}.
\end{align*}
We can bound $\left\Vert \mathbb{E}\left[\sum_{i,j}\s_{(ij)}\s_{(ij)}^{\top}\right]\right\Vert $
by the same quantity in a similar manner. Treating $\{\s_{(ij)}\}$
as $n\times1$ matrices and applying the matrix Bernstein inequality
in Theorem~\ref{lem:matrix_bernstein} gives that w.h.p. 
\[
\left\Vert \left(\left(\P_{T}\R_{\Omega}-\P_{T}\right)\Z\right)\e_{b}\right\Vert _{2}\le\frac{1}{2}\sqrt{\frac{n}{\mu_{0}r}}\left\Vert \Z\right\Vert _{\infty}+\frac{1}{2}\left\Vert \Z\right\Vert _{\infty,2}
\]
provided $c_{0}$ in the lemma statement is large enough. In a similar
fashion we prove that $\left\Vert \e_{a}^{\top}\left(\left(\P_{T}\R_{\Omega}-\P_{T}\right)\Z\right)\right\Vert $
is bounded by the same quantity w.h.p. The lemma follows from a union
bound over all $(a,b)\in[n]\times[n]$.

\section{Proof of Corollary~\ref{cor:SVDproj}\label{sec:proof_SVDproj}}

When $p\gtrsim\frac{\log^{2}n}{n}$, the standard Bernstein inequality
and a union bound implies that w.h.p. the degrees (i.e., the number
of observed entries) of the rows and columns of $\P_{\Omega}\M$ are
bounded by $2pn$. This means $\frac{1}{p}\widetilde{\M}^{\Omega}=\frac{1}{p}\P_{\Omega}\M=\R_{\Omega}\M$.
By Lemma~\ref{lem:inf}, we have

\begin{equation}
\left\Vert \frac{1}{p}\widetilde{\M}^{\Omega}-\M\right\Vert \le c\left(\frac{1}{p}\left\Vert \M\right\Vert _{\infty}\log n+\sqrt{\frac{1}{p}\log n}\left\Vert \M\right\Vert _{\infty,2}\right).\label{eq:k1}
\end{equation}
Let $\sigma_{i}$ be the $i$-th singular value of $\M$ (with $\sigma_{i}=0$
for $i>r$), and recall that $\tilde{\sigma}_{i}$ is the $i$-th
singular values of $\widetilde{\M}^{\Omega}$. By Weyl's inequality~\cite{bhatia1987perturbation},
we obtain that for $i=r+1,\ldots,n$, 
\begin{equation}
\frac{1}{p}\tilde{\sigma}_{i}=\left|\frac{1}{p}\tilde{\sigma}_{i}-\sigma_{i}\right|\le\left\Vert \frac{1}{p}\widetilde{\M}^{\Omega}-\M\right\Vert .\label{eq:k2}
\end{equation}
It follows that 
\begin{align*}
\left\Vert \M-\textsf{T}_{r}(\widetilde{\M}^{\Omega})\right\Vert  & \le\left\Vert \M-\frac{1}{p}\widetilde{\M}^{\Omega}\right\Vert +\left\Vert \frac{1}{p}\widetilde{\M}^{\Omega}-\textsf{T}_{r}(\widetilde{\M}^{\Omega})\right\Vert \\
 & =\left\Vert \M-\frac{1}{p}\widetilde{\M}^{\Omega}\right\Vert +\max_{i=r+1,\ldots,n}\frac{1}{p}\tilde{\sigma}_{i}\\
 & \le2c\left(\frac{1}{p}\left\Vert \M\right\Vert _{\infty}\log n+\sqrt{\frac{1}{p}\log n}\left\Vert \M\right\Vert _{\infty,2}\right),
\end{align*}
where we use~(\ref{eq:k1}) and~(\ref{eq:k2}) in the last inequality.
Since the rank of $\M-\textsf{T}_{r}(\widetilde{\M}^{\Omega})$ is
at most $r$, we have $\left\Vert \M-\textsf{T}_{r}(\widetilde{\M}^{\Omega})\right\Vert _{F}\le\sqrt{r}\left\Vert \M-\textsf{T}_{r}(\widetilde{\M}^{\Omega})\right\Vert $
and the corollary follows.

\section{Proof of Theorem~\ref{thm:cluster}\label{sec:proof_cluster}}

The proof is similar to that of Theorem~\ref{thm:uniform}, and we
shall point out where they differ. We use the same notations as in
the proof of Theorem~\ref{thm:uniform}, except that throughout this
section we \emph{re-define} the two projections: 
\begin{align*}
\P_{T}\Z & :=\U\U^{\top}\Z\bar{\V}\bar{\V}^{\top}+\bar{\U}\bar{\U}^{\top}\Z\V\V^{\top}-\U\U^{\top}\Z\V\V^{\top},\\
\P_{T^{\bot}}\Z & :=(\bar{\U}\bar{\U}^{\top}-\U\U^{\top})\Z(\bar{\V}\bar{\V}^{\top}-\V\V^{\top}).
\end{align*}
Note that $\P_{T}\Z+\P_{T^{\bot}}\Z=\bu\bu^{\top}\Z\bv\bv^{\top}$.
Since $\mbox{col}(\U)\subseteq\mbox{col}(\bu)$ and $\frac{\mu_{0}r}{n}\le\frac{\bar{\mu}_{0}\bar{r}}{n}$,
one can verify that under the incoherence assumption on $\U$ and
$\bu$ in the theorem statement, we have for all $i,j,b\in[n]$, 
\begin{align}
\left\Vert \P_{T}(\ei\ej^{\top})\right\Vert _{F}^{2} & =\left\Vert \U^{\top}\ei\ej^{\top}\bar{\V}\right\Vert _{F}^{2}+\left\Vert \bu^{\top}\ei\ej^{\top}\V\right\Vert _{F}^{2}-\left\Vert \U^{\top}\ei\ej^{\top}\V\right\Vert _{F}^{2}\le2\frac{\mu_{0}r}{n}\cdot\frac{\bar{\mu}_{0}\bar{r}}{n}.\label{eq:basic_ineq_cluster}\\
\left\Vert \P_{T^{\bot}}(\ei\ej^{\top})\right\Vert _{F} & =\left\Vert (\bar{\U}\bar{\U}^{\top}-\U\U^{\top})\ei\right\Vert _{2}\left\Vert \ej^{\top}(\bar{\V}\bar{\V}^{\top}-\V\V^{\top})\right\Vert _{2}\le\frac{\bar{\mu}_{0}\bar{r}}{n}\label{eq:basic_ineq_cluster2}\\
\left\Vert \P_{T}(\ei\ej^{\top})\e_{b}\right\Vert _{2}^{2} & \le2\frac{\mu_{0}r}{n}\frac{\bar{\mu}_{0}^{2}\bar{r}^{2}}{n}.\label{eq:basic_ineq_cluster3}
\end{align}
We have the following subgradient optimality condition.
\begin{prop}
\label{prop:opt_cond_cluster} $\X^{*}:=\bar{\U}^{\top}\M\bar{\V}$
is the unique optimal solution to the program~\eqref{eq:cluster_prog}
if the following conditions are satisfied: 1. $\left\Vert \P_{T}-\P_{T}\R_{\Omega}\P_{T}\right\Vert _{op}\le\frac{1}{2}$
and $\frac{1}{\sqrt{p}}\left\Vert \P_{\Omega}\P_{T^{\bot}}\right\Vert _{op}\le\sqrt{\frac{2\bar{\mu}_{0}\bar{r}}{\mu_{0}r}}$;
2. there exist a dual certificate $\Y$ with $\P_{\Omega}\Y=\Y$ and
obeys (a) $\left\Vert \P_{T}\Y-\U\V^{\top}\right\Vert _{F}\le\sqrt{\frac{\mu_{0}r}{32\bar{\mu}_{0}\bar{r}}}$
and (b) $\left\Vert \P_{T^{\bot}}\Y\right\Vert \le\frac{1}{2}.$\end{prop}
\begin{proof}
Consider any feasible solution $\X$ to~\eqref{eq:cluster_prog}.
Let $\bDelta:=\bar{\U}\X\bar{\V}^{\top}-\M$ and $\G\in\mathbb{R}^{n\times n}$
be such that $\left\Vert \P_{T^{\bot}}\G\right\Vert =1$ and $\left\langle \P_{T^{\bot}}\G,\P_{T^{\bot}}\bDelta\right\rangle =\left\Vert \P_{T^{\bot}}\bDelta\right\Vert _{*}$.
Note that $\P_{\Omega}\bDelta=0$, $\left\Vert \X\right\Vert _{*}=\left\Vert \bu\X\bar{\V}^{\top}\right\Vert _{*}$
and $\left\Vert \bu^{\top}\M\bar{\V}\right\Vert _{*}=\left\Vert \M\right\Vert _{*}$.
Similarly to the proof of Proposition~\ref{prop:opt_cond}, we have
\begin{equation}
\left\Vert \X\right\Vert _{*}-\left\Vert \bu^{\top}\M\bar{\V}\right\Vert _{*}\ge\left\langle \U\V^{\top}-\P_{T}\Y+\P_{T^{\bot}}\G-\P_{T^{\bot}}\Y,\bDelta\right\rangle \ge\frac{1}{2}\left\Vert \P_{T^{\bot}}\bDelta\right\Vert _{*}-\sqrt{\frac{\mu_{0}r}{32\bar{\mu}_{0}\bar{r}}}\left\Vert \P_{T}\bDelta\right\Vert _{F}.\label{eq:subgrad_inq_cluster}
\end{equation}
On the other hand, since $\left\Vert \P_{T}\R_{\Omega}\P_{T}-\P_{T}\right\Vert _{op}\le\frac{1}{2}$
by assumption, we have 
\[
\frac{1}{\sqrt{2}}\left\Vert \P_{T}\bDelta\right\Vert _{F}\le\sqrt{\left\langle \left(\P_{T}\R_{\Omega}\P_{T}-\P_{T}\right)\bDelta,\P_{T}\bDelta\right\rangle +\left\langle \P_{T}\bDelta,\P_{T}\bDelta\right\rangle }=\frac{1}{\sqrt{p}}\left\Vert \P_{\Omega}\P_{T}\bDelta\right\Vert _{F}.
\]
Because $\bu\bu^{\top}\bDelta\bv\bv^{\top}=\bDelta$, we have $0=\P_{\Omega}(\bDelta)=\P_{\Omega}\left(\P_{T}+\P_{T^{\bot}}\right)\bDelta$
and thus 
\[
\frac{1}{\sqrt{p}}\left\Vert \P_{\Omega}\P_{T}\bDelta\right\Vert _{F}=\frac{1}{\sqrt{p}}\left\Vert \P_{\Omega}\P_{T^{\bot}}\bDelta\right\Vert _{F}\le\sqrt{\frac{2\bar{\mu}_{0}\bar{r}}{\mu_{0}r}}\left\Vert \P_{T^{\bot}}\bDelta\right\Vert _{F},
\]
where the last inequality follows from Condition 1 in the statement
of the proposition. Combining the last two display equations gives
$\left\Vert \P_{T}\bDelta\right\Vert _{F}\le\sqrt{\frac{2\bar{\mu}_{0}\bar{r}}{\mu_{0}r}}\left\Vert \P_{T^{\bot}}\bDelta\right\Vert _{*}.$
It follows from~\eqref{eq:subgrad_inq_cluster} that 
\begin{align*}
\left\Vert \X\right\Vert _{*}-\left\Vert \bu^{\top}\M\bv\right\Vert _{*} & \ge\frac{1}{2}\left\Vert \P_{T^{\bot}}\bDelta\right\Vert _{*}-\frac{1}{4}\left\Vert \P_{T^{\bot}}\bDelta\right\Vert _{*}\ge\frac{1}{4}\left\Vert \P_{T^{\bot}}\bDelta\right\Vert _{*}.
\end{align*}
The last RHS is strictly positive for all $\bDelta$ with $\P_{\Omega}\bDelta=0$
and $\bDelta\neq0$; otherwise we would have $\P_{T}\bDelta=\left(\P_{T}+\P_{T^{\bot}}\right)\bDelta=\bDelta$
and thus $\P_{T}\R_{\Omega}\P_{T}\bDelta=0$, contradicting $\left\Vert \P_{T}\R_{\Omega}\P_{T}-\P_{T}\right\Vert _{op}\le\frac{1}{2}$.
This proves that $\X^{*}:=\bu^{\top}\M\bv$ is the unique optimal
solution to~\eqref{eq:cluster_prog}. 
\end{proof}
We proceed by showing that Condition 1 in Proposition~\ref{prop:opt_cond_cluster}
is satisfied w.h.p. under the conditions of Theorem~\ref{thm:cluster}.
This is done in the lemma below, which is proved in Section~\ref{sub:proof_op_cluster}
to follow.
\begin{lem}
\label{lem:op_cluster}If $p\ge c_{0}\frac{\mu_{0}\bar{\mu}_{0}r\bar{r}}{n^{2}}\log n$
for some sufficiently large constant $c_{0}$, then w.h.p. we have
\[
\left\Vert \P_{T}\R_{\Omega}\P_{T}-\P_{T}\right\Vert _{op}\le\frac{1}{2}\qquad\text{and}\qquad\frac{1}{\sqrt{p}}\left\Vert \P_{\Omega}\P_{T^{\bot}}\right\Vert _{op}\le\sqrt{\frac{2\bar{\mu}_{0}\bar{r}}{\mu_{0}r}}.
\]

\end{lem}
We now construct a dual certificate $\Y$ using the golfing scheme.
This is done similarly as before; in particular, we let $k_{0}:=20\log(32\bar{\mu}_{0}\bar{r})$,
$q:=1-(1-p)^{1/k_{0}}\ge\frac{p}{k_{0}}$, $\W_{k}$ be given by~\eqref{eq:W_defn}
(with the re-defined $\P_{T}$) and $\Y:=\W_{k_{0}}$. Clearly $\P_{\Omega}(\Y)=\Y$
by construction. Note that for $k\in[k_{0}]$, the matrix $\D_{k}:=\U\V^{\top}-\P_{T}(\W_{k})$
again satisfies~\eqref{eq:recursion}. It follows that $\left\Vert \D_{k}\right\Vert _{F}\le\frac{1}{2}\left\Vert \D_{k-1}\right\Vert _{F}$
by the first inequality in Lemma~\ref{lem:op_cluster} and thus 
\[
\left\Vert \P_{T}\Y-\U\V^{\top}\right\Vert _{F}=\left\Vert \D_{k_{0}}\right\Vert _{F}\le\left(\frac{1}{2}\right)^{k_{0}}\left\Vert \D_{0}\right\Vert _{F}\le\sqrt{\frac{r}{32\bar{\mu}_{0}\bar{r}}}\le\sqrt{\frac{\mu_{0}r}{32\bar{\mu}_{0}\bar{r}}},
\]
proving Condition 2(a) in Proposition~\ref{prop:opt_cond_cluster}.
To prove Condition 2(b), we need three lemmas which are analogues
of Lemmas~\ref{lem:op_inf},~\ref{lem:inf_2} and~\ref{lem:inf}
in the proof of Theorem~\ref{thm:uniform}.
\begin{lem}
\label{lem:op_inf_cluster} Suppose $\Z$ is a fixed $n\times n$
matrix. For some universal constant $c>1$, we have w.h.p. 
\[
\left\Vert \P_{T^{\bot}}\left(\R_{\Omega}-\II\right)\Z\right\Vert \le c\left(\frac{\bar{\mu}_{0}\bar{r}}{pn}\log n\left\Vert \Z\right\Vert _{\infty}+\sqrt{\frac{\bar{\mu}_{0}\bar{r}\log n}{pn}}\left\Vert \Z\right\Vert _{\infty,2}\right).
\]

\end{lem}

\begin{lem}
\label{lem:inf2_cluster}Suppose $\Z$ is a fixed matrix. If $p\ge c_{0}\frac{\mu_{0}\bar{\mu}_{0}r\bar{r}\log n}{n}$
for some $c_{0}$ sufficiently large, then w.h.p. 
\[
\left\Vert \left(\P_{T}\R_{\Omega}-\P_{T}\right)\Z\right\Vert _{\infty,2}\le\frac{1}{2}\sqrt{\frac{n}{\mu_{0}r}}\left\Vert \Z\right\Vert _{\infty}+\frac{1}{2}\left\Vert \Z\right\Vert _{\infty,2}
\]

\end{lem}

\begin{lem}
\label{lem:inf_cluster} Suppose $\Z$ is a fixed $n\times n$ matrix.
If $p\ge c_{0}\frac{\mu_{0}\bar{\mu}_{0}r\bar{r}\log n}{n}$ for some
$c_{0}$ sufficiently large, then w.h.p. 
\[
\left\Vert \left(\P_{T}\R_{\Omega}-\P_{T}\right)\Z\right\Vert _{\infty}\le\frac{1}{2}\left\Vert \Z\right\Vert _{\infty}.
\]

\end{lem}
We prove these lemmas in Sections~\ref{sub:proof_op_inf_cluster}--\ref{sub:proof_inf_cluster}
to follow. Following the same lines as in the proof of Theorem~\ref{thm:uniform},
we obtain 
\[
\left\Vert \P_{T^{\bot}}\Y\right\Vert \le\sum_{k=1}^{k_{0}}\left\Vert \P_{T^{\bot}}\left(\R_{\Omega_{k}}\P_{T}-\P_{T}\right)\D_{k-1}\right\Vert .
\]
Applying Lemma~\ref{lem:op_inf_cluster} with $\Omega$ replaced
by $\Omega_{k}$ to each summand of the last R.H.S, we get that w.h.p.
\begin{align*}
\left\Vert \P_{T^{\bot}}\Y\right\Vert  & \le c\frac{\bar{\mu}_{0}\bar{r}\log n}{qn}\sum_{k=1}^{k_{0}}\left\Vert \D_{k-1}\right\Vert _{\infty}+c\sqrt{\frac{\bar{\mu}_{0}\bar{r}\log n}{qn}}\sum_{k=1}^{k_{0}}\left\Vert \D_{k-1}\right\Vert _{\infty,2}\\
 & \le\frac{c'}{\sqrt{c_{0}}}\frac{n}{\mu_{0}r}\sum_{k=1}^{k_{0}}\left\Vert \D_{k-1}\right\Vert _{\infty}+\frac{c'}{\sqrt{c_{0}}}\sqrt{\frac{n}{\mu_{0}r}}\sum_{k=1}^{k_{0}}\left\Vert \D_{k-1}\right\Vert _{\infty,2}
\end{align*}
where the last inequality follows from $q\ge\frac{p}{20\log(32\bar{\mu}_{0}\bar{r})}\ge c_{0}\frac{\mu_{0}\bar{\mu}_{0}r\bar{r}\log n}{20n^{2}}.$
Again following the same lines as in the proof of Theorem~\ref{thm:uniform},
but using the new Lemmas \ref{lem:inf_cluster} and~\ref{lem:inf2_cluster},
we can bound the two terms above as 
\begin{align*}
\left\Vert \D_{k-1}\right\Vert _{\infty} & \le\left(\frac{1}{2}\right)^{k-1}\left\Vert \U\V^{\top}\right\Vert _{\infty},\\
\left\Vert \D_{k-1}\right\Vert _{\infty,2} & \le k\left(\frac{1}{2}\right)^{k-1}\sqrt{\frac{n}{\mu_{0}r}}\left\Vert \U\V^{\top}\right\Vert _{\infty}+\left(\frac{1}{2}\right)^{k-1}\left\Vert \U\V^{\top}\right\Vert _{\infty,2}.
\end{align*}
It follows that 
\[
\left\Vert \P_{T^{\bot}}\Y\right\Vert \le\frac{c''}{\sqrt{c_{0}}}\left(\frac{n}{\mu_{0}r}\left\Vert \U\V^{\top}\right\Vert _{\infty}+\sqrt{\frac{n}{\mu_{0}r}}\left\Vert \U\V^{\top}\right\Vert _{\infty,2}\right).
\]
The inequality $\left\Vert \P_{T^{\bot}}(\Y)\right\Vert \le\frac{1}{2}$
then follows from the incoherence conditions~\eqref{eq:incoherence}
of $\U$ and $\V$ provided $c_{0}$ is sufficiently large. This proves
Condition 2(b) in Proposition~\ref{prop:opt_cond_cluster} and hence
completes the proof of Theorem~\ref{thm:cluster}.

\subsection{Proof of Lemma~\ref{lem:op_cluster}\label{sub:proof_op_cluster}}

The proof of the first inequality is identical to that of Lemma~\ref{lem:op}
except that we use~\eqref{eq:basic_ineq_cluster} instead of~\eqref{eq:basic_inequality}
(cf. Theorem 4.1 in~\cite{candes2009exact} and Lemma 11 in~\cite{chen2011LSarxiv}). 

To prove the second inequality, recall that $\gamma_{ij}:=\mathbb{I}\left((i,j)\in\Omega\right)$
is the indicator variable. For $(i,j)\in[n]\times[n]$, let $\mathcal{S}_{(ij)}$
be the operator that maps $\Z\in\mathbb{R}^{n\times n}$ to $\left(\frac{1}{p}\gamma_{ij}-1\right)\left\langle \P_{T^{\bot}}\Z,\ei\ej^{\top}\right\rangle \P_{T^{\bot}}\left(\ei\ej^{\top}\right)$.
Observe that $\left\{ \mathcal{S}_{(ij)}\right\} $ are independent
zero-mean self-adjoint operators, and 
\[
\frac{1}{p}\mathcal{P}_{T^{\bot}}\mathcal{P}_{\Omega}\mathcal{P}_{T^{\bot}}-\mathcal{P}_{T^{\bot}}=\sum_{i,j}\mathcal{S}_{(ij)}.
\]
By~\eqref{eq:basic_ineq_cluster2}, we know that for any $\Z$,
\[
\left\Vert \mathcal{S}_{(ij)}\Z\right\Vert _{F}\le\frac{1}{p}\left\Vert \Z\right\Vert _{F}\left\Vert \P_{T^{\bot}}\left(\ei\ej^{\top}\right)\right\Vert _{F}^{2}\le\frac{\bar{\mu}_{0}^{2}\bar{r}^{2}}{pn^{2}}\left\Vert \Z\right\Vert _{F},
\]
and
\begin{align*}
\left\Vert \mathbb{E}{\textstyle \sum_{i,j}}\mathcal{S}_{(ij)}^{2}\Z\right\Vert _{F} & =\left\Vert \mathbb{E}\sum_{i,j}\left(\frac{1}{p}\gamma_{ij}-1\right)^{2}\left\langle \P_{T^{\bot}}\Z,\ei\ej^{\top}\right\rangle \left\Vert \P_{T^{\bot}}\left(\ei\ej^{\top}\right)\right\Vert _{F}^{2}\P_{T^{\bot}}\left(\ei\ej^{\top}\right)\right\Vert _{F}\\
 & \le\frac{1-p}{p}\left(\max_{i,j}\left\Vert \P_{T^{\bot}}\left(\ei\ej^{\top}\right)\right\Vert _{F}^{2}\right)\left\Vert {\textstyle \sum_{i,j}}\left\langle \P_{T^{\bot}}\Z,\ei\ej^{\top}\right\rangle \ei\ej^{\top}\right\Vert _{F}\le\frac{\bar{\mu}_{0}^{2}\bar{r}_{0}^{2}}{pn^{2}}\left\Vert \mathcal{P}_{T^{\bot}}\Z\right\Vert _{F}.
\end{align*}
This means 
\[
\left\Vert \mathcal{S}_{(ij)}\right\Vert _{op}\le\frac{\bar{\mu}_{0}^{2}\bar{r}^{2}}{pn^{2}}\qquad\text{and}\qquad\left\Vert \mathbb{E}{\textstyle \sum_{i,j}}\mathcal{S}_{(ij)}^{2}\right\Vert _{op}\le\frac{\bar{\mu}_{0}^{2}\bar{r}_{0}^{2}}{pn^{2}}.
\]
Applying the the matrix Bernstein inequality in Theorem~\ref{lem:matrix_bernstein},
we obtain w.h.p.
\[
\left\Vert \frac{1}{p}\mathcal{P}_{T^{\bot}}\mathcal{P}_{\Omega}\mathcal{P}_{T^{\bot}}-\mathcal{P}_{T^{\bot}}\right\Vert _{op}\le c\frac{\bar{\mu}_{0}^{2}\bar{r}^{2}}{pn^{2}}\log\left(2n\right)+\frac{\bar{\mu}_{0}\bar{r}}{\sqrt{p}n}\sqrt{4c\log(2n)}\le c'\frac{\bar{\mu}_{0}\bar{r}\sqrt{\bar{\mu}_{0}\bar{r}\log(2n)}}{\sqrt{p\mu_{0}r}n}
\]
for some constant $c'$, where the last inequality follows from $\mu_{0}r\le\bar{\mu}_{0}\bar{r}$
and the assumption $p\ge c_{0}\frac{\mu_{0}\bar{\mu}_{0}r\bar{r}}{n^{2}}\log(2n)$.
It follows that
\begin{align*}
\frac{1}{\sqrt{p}}\left\Vert \mathcal{P}_{\Omega}\mathcal{P}_{T^{\bot}}\right\Vert _{op}\le\sqrt{\left\Vert \frac{1}{p}\mathcal{P}_{T^{\bot}}\mathcal{P}_{\Omega}\mathcal{P}_{T^{\bot}}\right\Vert _{op}} & \le\sqrt{\left\Vert \frac{1}{p}\mathcal{P}_{T^{\bot}}\mathcal{P}_{\Omega}\mathcal{P}_{T^{\bot}}-\mathcal{P}_{T^{\bot}}\right\Vert _{op}+\left\Vert \mathcal{P}_{T^{\bot}}\right\Vert _{op}}\\
 & \le\sqrt{c'\sqrt{\frac{\bar{\mu}_{0}^{3}\bar{r}^{3}\log(2n)}{p\mu_{0}rn^{2}}}+1}\le\sqrt{\frac{2\bar{\mu}_{0}\bar{r}}{\mu_{0}r}},
\end{align*}
where in the last inequality we again use the assumption $p\ge c_{0}\frac{\mu_{0}\bar{\mu}_{0}r\bar{r}}{n^{2}}\log(2n).$

\subsection{Proof of Lemma~\ref{lem:op_inf_cluster}\label{sub:proof_op_inf_cluster}}

We can write 
\[
\P_{T^{\bot}}\left(\R_{\Omega}-\II\right)\Z=\sum_{i,j}\S_{(ij)}:=\sum_{i,j}\left(\frac{1}{p}\gamma_{ij}-1\right)Z_{ij}\P_{T^{\bot}}\left(\ei\ej^{\top}\right),
\]
where $\left\{ \S_{(ij)}\right\} $ are independent $n\times n$ matrices
satisfying $\mathbb{E}[\S_{(ij)}]=0$ and 
\[
\left\Vert \S_{(ij)}\right\Vert \le\frac{1}{p}\left|Z_{ij}\right|\left\Vert \P_{T^{\bot}}(\ei\ej^{\top})\right\Vert _{F}\le\frac{\bar{\mu}_{0}\bar{r}}{pn}\left\Vert \Z\right\Vert _{\infty}.
\]
by~\eqref{eq:basic_ineq_cluster2}. Moreover, since $\mbox{col}(\U)\subseteq\mbox{col}(\bar{\U})$,
$\mbox{col}(\V)\subseteq\mbox{col}(\bv)$ and $\bu\bu^{\top}-\U\U^{\top}$,
$\bv\bv^{\top}-\V\V^{\top}$ are projections, we have 
\[
\left(\P_{T^{\bot}}\left(\ei\ej^{\top}\right)\right)^{\top}\P_{T^{\bot}}\left(\ei\ej^{\top}\right)=\left|\ei^{\top}\left(\bar{\U}\bar{\U}^{\top}-\U\U^{\top}\right)\ei\right|\left(\bv\bv^{\top}-\V\V^{\top}\right)\ej\ej^{\top}\left(\bv\bv^{\top}-\V\V^{\top}\right)
\]
and thus 
\begin{align*}
\left\Vert \mathbb{E}{\textstyle \sum_{i,j}}\S_{(ij)}^{\top}\S_{(ij)}\right\Vert = & \frac{1-p}{p}\left\Vert \sum_{j}\left(\bv\bv^{\top}-\V\V^{\top}\right)\ej\ej^{\top}\left(\bv\bv^{\top}-\V\V^{\top}\right)\sum_{i}\left|\ei^{\top}\left(\bar{\U}\bar{\U}^{\top}-\U\U^{\top}\right)\ei\right|Z_{ij}^{2}\right\Vert \\
\le & \frac{1-p}{p}\left\Vert \sum_{j}\ej\ej^{\top}\sum_{i}\left|\ei^{\top}\left(\bar{\U}\bar{\U}^{\top}-\U\U^{\top}\right)\ei\right|Z_{ij}^{2}\right\Vert \le\frac{\bar{\mu}_{0}\bar{r}}{pn}\left\Vert \Z\right\Vert _{\infty,2}^{2}.
\end{align*}
We can bound $\left\Vert \mathbb{E}\sum_{i,j}\S_{(ij)}\S_{(ij)}^{\top}\right\Vert $
by the same quantity in a similar manner. Applying the matrix Bernstein
inequality in Theorem~\ref{lem:matrix_bernstein} proves the lemma.

\subsection{Proof of Lemma~\ref{lem:inf2_cluster}}

Fix $b\in[n]$. The $b$-th column of the matrix $\left(\P_{T}\R_{\Omega}-\P_{T}\right)\Z$
can be written as 
\[
\left(\left(\P_{T}\R_{\Omega}-\P_{T}\right)\Z\right)\e_{b}=\sum_{i,j}\s_{(ij)}:=\sum_{i,j}\left(\frac{1}{p}\gamma_{ij}-1\right)Z_{ij}\P_{T}(\ei\ej^{\top})\e_{b},
\]
where $\{\s_{(ij)}\}$ are independent column vectors in $\mathbb{R}^{n}$.
Note that $\mathbb{E}\left[\s_{(ij)}\right]=0$ and

\[
\left\Vert \s_{(ij)}\right\Vert _{2}\le\frac{1}{p}\left|Z_{ij}\right|\left\Vert \P_{T}(\ei\ej^{\top})\e_{b}\right\Vert _{2}\le\frac{2}{p}\frac{\bar{\mu}_{0}\bar{r}}{n}\sqrt{\frac{\mu_{0}r}{n}}\left\Vert \Z\right\Vert _{\infty}\le\frac{2}{c_{0}\log n}\sqrt{\frac{n}{\mu_{0}r}}\left\Vert \Z\right\Vert _{\infty}
\]
by~\eqref{eq:basic_ineq_cluster3} and the assumption on $p$. We
also have
\[
\left|\mathbb{E}\sum_{i,j}\s_{(ij)}^{\top}\s_{(ij)}\right|=\left|\sum_{i,j}\mathbb{E}\left[\left(\frac{1}{p}\gamma_{ij}-1\right)^{2}\right]Z_{ij}^{2}\left\Vert \P_{T}(\ei\ej^{\top})\e_{b}\right\Vert _{2}^{2}\right|=\frac{1-p}{p}\sum_{i,j}Z_{ij}^{2}\left\Vert \P_{T}(\ei\ej^{\top})\e_{b}\right\Vert _{2}^{2}.
\]
Because 
\begin{align*}
\left\Vert \P_{T}(\ei\ej^{\top})\e_{b}\right\Vert _{2} & =\left\Vert \U\U^{\top}\ei\ej^{\top}\bv\bv^{\top}\e_{b}+(\bu\bu^{\top}-\U\U^{\top})\ei\ej^{\top}\V\V^{\top}\e_{b}\right\Vert _{2}\\
 & \le\sqrt{\frac{\mu_{0}r}{n}}\left|\ej^{\top}\bv\bv^{\top}\e_{b}\right|+\sqrt{\frac{\bar{\mu}_{0}\bar{r}}{n}}\left|\ej^{\top}\V\V^{\top}\e_{b}\right|,
\end{align*}
it follows that 
\begin{align*}
\left|\mathbb{E}\sum_{i,j}\s_{(ij)}^{\top}\s_{(ij)}\right| & \le\frac{2}{p}\sum_{i,j}Z_{ij}^{2}\frac{\mu_{0}r}{n}\left|\ej^{\top}\bv\bv^{\top}\e_{b}\right|^{2}+\frac{2}{p}\sum_{i,j}Z_{ij}^{2}\frac{\bar{\mu}_{0}\bar{r}}{n}\left|\ej^{\top}\V\V^{\top}\e_{b}\right|^{2}\\
 & =\frac{2}{p}\frac{\mu_{0}r}{n}\sum_{j}\left|\ej^{\top}\bv\bv^{\top}\e_{b}\right|^{2}\sum_{i}Z_{ij}^{2}+\frac{2}{p}\frac{\bar{\mu}_{0}\bar{r}}{n}\sum_{j}\left|\ej^{\top}\V\V^{\top}\e_{b}\right|^{2}\sum_{i}Z_{ij}^{2}\\
 & \le\frac{2}{p}\frac{\mu_{0}r}{n}\left\Vert \Z\right\Vert _{\infty,2}^{2}\sum_{j}\left|\ej^{\top}\bv\bv^{\top}\e_{b}\right|^{2}+\frac{2}{p}\frac{\bar{\mu}_{0}\bar{r}}{n}\left\Vert \Z\right\Vert _{\infty,2}^{2}\sum_{j}\left|\ej^{\top}\V\V^{\top}\e_{b}\right|^{2}\\
 & =\frac{2}{p}\frac{\mu_{0}r}{n}\left\Vert \Z\right\Vert _{\infty,2}^{2}\left\Vert \bv^{\top}\e_{b}\right\Vert _{2}^{2}+\frac{2}{p}\frac{\bar{\mu}_{0}\bar{r}}{n}\left\Vert \Z\right\Vert _{\infty,2}^{2}\left\Vert \V^{\top}\e_{b}\right\Vert _{2}^{2}\\
 & \le\frac{4}{p}\frac{\mu_{0}r\bar{\mu}_{0}\bar{r}}{n^{2}}\left\Vert \Z\right\Vert _{\infty,2}^{2}\le\frac{4}{c_{0}\log n}\left\Vert \Z\right\Vert _{\infty,2}^{2}.
\end{align*}
We can bound $\left\Vert \mathbb{E}\left[\sum_{i,j}\s_{(ij)}^{\top}\s_{(ij)}\right]\right\Vert $
in a similar manner. Treating $\{\S_{(ij)}\}$ as $n\times1$ matrices
and applying the Matrix Bernstein inequality in Theorem~\ref{lem:matrix_bernstein},
we get 
\[
\left\Vert \left(\left(\P_{T}\R_{\Omega}-\P_{T}\right)\Z\right)\e_{b}\right\Vert _{2}\le\frac{1}{2}\sqrt{\frac{n}{\mu_{0}r}}\left\Vert \Z\right\Vert _{\infty}+\frac{1}{2}\left\Vert \Z\right\Vert _{\infty,2},\quad\textrm{w.h.p.}
\]
provided $c_{0}$ in the statement of the lemma is sufficiently large.
In a similar fashion we can prove that the $\left\Vert \e_{a}^{\top}\left(\left(\P_{T}\R_{\Omega}-\P_{T}\right)\Z\right)\right\Vert _{2}$
is bounded by the same quantity w.h.p. The lemma follows from a union
bound over all $(a,b)\in[n]\times[n]$.

\subsection{Proof of Lemma~\ref{lem:inf_cluster}\label{sub:proof_inf_cluster}}

Fix $(a,b)\in[n]\times[n]$. We can write the $(a,b)$ entry of the
matrix $\left(\P_{T}\R_{\Omega}-\P_{T}\right)\Z$ as 
\begin{align*}
\left[\left(\P_{T}\R_{\Omega}-\P_{T}\right)\Z\right]_{ab} & =\sum_{i,j}s_{ij}:=\sum_{i,j}\left(\frac{1}{p}\gamma_{ij}-1\right)Z_{ij}\left\langle \ei\ej^{\top},\P_{T}(\e_{a}\e_{b}^{\top})\right\rangle ,
\end{align*}
where $s_{ij}\in\mathbb{R}$ are independent zero-mean random variables.
By~\eqref{eq:basic_ineq_cluster} and the assumption on $p$, we
have 
\[
\left|s_{ij}\right|\le\frac{1}{p}\left|Z_{ij}\right|\left\Vert \P_{T}(\ei\ej^{\top})\right\Vert _{F}\left\Vert \P_{T}(\e_{a}\e_{b}^{\top})\right\Vert _{F}\le\frac{1}{2c_{0}\log n}\left\Vert \Z\right\Vert _{\infty}.
\]
and
\begin{align*}
\left|\mathbb{E}\sum_{i,j}s_{ij}^{2}\right| & =\sum_{i,j}\mathbb{E}\left[\left(\frac{1}{p}\gamma_{ij}-1\right)^{2}\right]Z_{ij}^{2}\left\langle \ei\ej^{\top},\P_{T}(\e_{a}\e_{b}^{\top})\right\rangle ^{2}\\
 & \le\frac{1}{p}\left\Vert \Z\right\Vert _{\infty}^{2}\sum_{i,j}\left\langle \ei\ej^{\top},\P_{T}(\e_{a}\e_{b}^{\top})\right\rangle ^{2}\\
 & =\frac{1}{p}\left\Vert \Z\right\Vert _{\infty}^{2}\left\Vert \P_{T}(\e_{a}\e_{b}^{\top})\right\Vert _{F}^{2}\le\frac{1}{2c_{0}\log n}\left\Vert \Z\right\Vert _{\infty}^{2}.
\end{align*}
Applying the Bernstein inequality in Theorem~\ref{lem:matrix_bernstein},
we conclude that w.h.p. $\left|\left[\left(\P_{T}\R_{\Omega}\P_{T}-\P_{T}\right)\Z\right]_{ab}\right|\le\frac{1}{2}\left\Vert \Z\right\Vert _{(\infty)}$
for $c_{0}$ sufficiently large. The lemma follows from a union bound
over all $(a,b)\in[n]\times[n]$.

\section{Proof of Theorem~\ref{thm:LS}\label{sec:proof_LS}}

\subsection{Part 1 of the theorem}

We first describe an equivalent formulation of the planted clique
problem. Let $\bar{\A}\in\mathbb{R}^{n\times n}$ be the adjacency
matrix of the graph, and $\L^{*}\in\left\{ 0,1\right\} ^{n\times n}$
be the matrix with $L_{ij}^{*}=1$ if and only if the nodes $i$ and
$j$ are both in the clique. Let $\bar{\S}^{*}:=\bar{\A}-\L^{*}$.
Note that for each $(i,j)\notin\mbox{support}(\L^{*})=\left\{ (i,j):L_{ij}^{*}=1\right\} $,
the pair $\bar{S}_{ij}^{*}=\bar{S}_{ji}^{*}$ is non-zero with probability
$1/2$; for each $(i,j)\in\mbox{support}(\L^{*})$, we always have
$\bar{S}_{ij}^{*}=\bar{S}_{ji}^{*}=0$.

We reduce the planted problem above to the matrix decomposition problem
using subsampling. Given the matrix $\bar{\A}$, we set each $\bar{A}_{ij}$
to zero with probability $\frac{2}{3}$ independently, and let $\A$
be the resulting matrix. If we let $\S^{*}:=\A-\L^{*}$, then each
pair $S_{ij}^{*}=S_{ji}^{*}$ is non-zero with probability $\tau=\frac{1}{3}$.
Moreover, the matrix $\L^{*}$ has rank $1$ and satisfies the standard
and joint incoherence conditions~(\ref{eq:incoherence}) and~(\ref{eq:strong_incoherence})
with parameters $\mu_{0}=1/n_{\min}$ and $\mu_{1}=n^{2}/n_{\min}^{2}$.
Hence recovering $\left(\L^{*},\S^{*}\right)$ from $\A$ is a special
case of the matrix decomposition problem. If there exists a polynomial-time
algorithm that, for all $n$, finds $\L^{*}$ given $\A$ with probability
at least $\frac{1}{2}$ when 
\[
\frac{\mu_{1}^{1-\epsilon'}}{n}=\frac{n^{1-2\epsilon'}}{n_{\min}^{2(1-\epsilon')}}\ge1,
\]
then it means this algorithm recovers the planted clique with $n_{\min}\le n^{\frac{1}{2}-\frac{\epsilon'}{2(1-\epsilon')}}$
from $\bar{\A}$, which violates the assumption \textbf{A1}.

\subsection{Part 2 of the theorem}

For simplicity, we assume $K:=\frac{n}{\mu_{0}r}$ and $ \frac{n}{2} $ are both integers. Let $ M := n/2 $.
Suppose $\L^{*}$ takes value uniformly at random from a set $\mathcal{L} =\left\{ \L^{(1)},\L^{(2)},\ldots,\L^{(M)}\right\} \subseteq\mathbb{R}^{n\times n}$
which we now define. Let $\L^{(0)}$ be the symmetric block-diagonal matrix 
with $r$ contingent blocks of size $K\times K$, where $L_{ij}^{(0)}=1$
inside the blocks and $0$ otherwise, and the blocks are in the first $ rK $ columns. Note that $ \mu_0 \ge 2 $ by assumption, so $ rK\le \frac{n}{2} $ and thus the last $ \frac{n}{2} $ rows and columns of $ \L^{(0)} $ are all zeros.
For $ l=1,\ldots, M $, let $ \L^{(l)} $ be the matrix obtained from $ \L^{(0)} $ by swapping the first row and column with the $ (n/2 + l) $-th row and column, respectively. In other words, the first block of $ \L^{(l)} $ corresponds to the rows and columns with indices $ \{2,3,\ldots,K, n/2+1\} $, and the other $ r-1 $ blocks are the same as those in $ \L^{(0)} .$ It is easy to check that each
$\L^{(l)}$ has rank $r$ and satisfies the standard incoherence condition~\eqref{eq:incoherence}
with parameter $\mu_{0}$.
We further assume that conditioned on $\L^{*}$, the matrix $ \S^* $ is distributed as follows: $S_{ij}^{*}$
equals $-1$ with probability $\tau=1/3$ and $0$ otherwise for $(i,j)\in\mbox{support}(\L^{*})$,
and $S_{ij}^{*}$ equals $1$ with probability $\tau=1/3$ and $0$
otherwise  for $(i,j)\not\in\mbox{support}(\L^{*})$. Finally, recall that $ \A = \L^* + \S^* $.

We now compute an upper bound on the mutual information $ I(\L^*; \A) $. Let $ \mathbb{P}^{(l)} $ be the distribution of $ \A $ conditioned on $ \L^* = \L^{(l)} $, and we use $ D\left(\mathbb{P}^{(l)} \Vert \mathbb{P}^{(l')} \right) $ to denote the  Kullback-Leibler (KL)  divergence between $ \mathbb{P}^{(l)}$ and $ \mathbb{P}^{(l')} $. By definition of the mutual information and the convexity of the KL divergence, we have
\begin{align*}
I(\L^*; \A) 
= \frac{1}{M} \sum_{l=1}^{M} D\left(\mathbb{P}^{(l)} \Vert \frac{1}{M} \sum_{l'=1}^{M} \mathbb{P}^{(l')} \right) 
\le \frac{1}{M^2} \sum_{l,l'=1}^{M} D\left(\mathbb{P}^{(l)} \Vert \mathbb{P}^{(l')} \right). 
\end{align*}
With slight abuse of notation, we use $ D(q_1\Vert q_2 ) : = q_1 \log \frac{q_1}{q_2} + (1-q_1) \log \frac{1-q_1}{1-q_2}$ to denote the KL divergence between two Bernoulli distributions with parameters $ q_1 $ and $ q_2 $. Direct computation gives
\begin{align*}
D\left(\mathbb{P}^{(l)} \Vert \mathbb{P}^{(l')} \right) 
= (K+1) D\left(\frac{2}{3} \Big \Vert \frac{1}{3} \right) + (K+1) D\left(\frac{1}{3} \Big \Vert \frac{2}{3} \right)
\le K+1
\end{align*}
for all $ l,l' = 1,\ldots, M $, where the inequality above follows from $ \log x \le x-1. $ It follows that $ I(\L^* ; \A)  \le K+1. $

We now apply the Fano's inequality~\cite{cover2012information} to obtain that for any measurable
function $\hat{\L}$ of $\A$,
\[
\mathbb{P}\left(\hat{\L}\neq\L^{*}\right)
\ge 1-\frac{I\left(\L^{*};\A\right)+\log 2}{\log M}
\ge 1-\frac{K+1+\log 2}{\log(n/2)}
\ge\frac{1}{2},
\]
where the probability is with respect to the randomness of $\L^{*}$
and $\S^{*}$, and the last inequality holds when $\frac{\log n}{12K} = \frac{\mu_{0}r \log n}{12n} \ge 1$ and $ n \ge 10 $.
Because the supremum is lower bounded by the average, we obtain 
\[
\sup_{\L^{*}\in\mathcal{L}}\mathbb{P}\left(\hat{\L}\neq\L^{*}\right)\ge\frac{1}{2},
\]
where the probability is with respect to the randomness of $\S^{*}$.

\bibliographystyle{plain}
\bibliography{incopt}

\end{document}